\documentclass[12pt]{article}

\usepackage{amsthm}
\usepackage{amsmath}
\usepackage{natbib}
\usepackage{fullpage}

\usepackage{chngcntr}

\usepackage[colorlinks,citecolor=blue,urlcolor=blue,filecolor=blue,backref=page]{hyperref}
\usepackage{graphicx}
\usepackage[bottom]{footmisc}
\usepackage{bm}
\usepackage{dsfont}
\usepackage{mathrsfs}
\usepackage[linesnumbered,ruled,vlined]{algorithm2e}
\usepackage{multirow}

\numberwithin{equation}{section}
\theoremstyle{plain}

\newtheorem{proposition}{Proposition}

\def\R{\mathds{R}}

\def\ga{\mbox{Ga}}
\def\iga{\mbox{IGa}}
\def\siga{\mbox{sqrtIGa}}
\def\be{\mbox{Be}}
\def\bin{\mbox{Bin}}
\def\geo{\mbox{Geo}}
\def\no{\mbox{N}}
\def\po{\mbox{Po}}
\def\pg{\mbox{Pg}}
\def\st{\mbox{St}}
\def\sno{\mbox{SkewN}}
\def\un{\mbox{Un}}
\def\bga{\mbox{BGa}}
\def\dir{\mbox{Dir}}
\def\tri{\mbox{Tri}}
\def\bebin{\mbox{BeBin}}
\def\dbb{\mbox{DoubleBB}}
\def\bbb{\mbox{BBB}}
\def\E{\mathds{E}}
\def\V{\mbox{Var}}
\def\Cov{\mbox{Cov}}
\def\Cor{\mbox{Corr}}
\def\p{\mbox{p}}
\def\b{\mbox{b}}
\def\P{\mbox{P}}
\def\Cr{\mbox{Corr}}
\def\Cv{\mbox{Cov}}
\def\pt{\mbox{PT}}
\def\rpt{\mbox{rPT}}
\def\dpt{\mbox{dPT}}
\def\bep{\mbox{BeP}}
\def\d{\mathrm{d}}
\def\rest{\mbox{rest}}
\def\data{\mbox{data}}
\def\tr{\mbox{tr}}
\def\bc{{\bm c}}
\def\bk{{\bm k}}
\def\bM{{\bm m}}
\def\bt{{\bm t}}
\def\bs{{\bm s}}
\def\bu{{\bm u}}
\def\bv{{\bm v}}
\def\bx{{\bm x}}
\def\by{{\bm y}}
\def\bz{{\bm z}}
\def\bC{{\bm C}}
\def\bS{{\bm S}}
\def\bT{{\bm T}}
\def\bU{{\bm U}}
\def\bV{{\bm V}}
\def\bW{{\bm W}}
\def\bX{{\bm X}}
\def\bY{{\bm Y}}
\def\bZ{{\bm Z}}
\def\bp{{\bm p}}
\def\bw{{\bm w}}
\def\bn{{\bm n}}
\def\bzero{{\bm 0}}
\def\simind{\stackrel{\mbox{\scriptsize{ind}}}{\sim}}
\def\simiid{\stackrel{\mbox{\scriptsize{iid}}}{\sim}}
\def\simas{\stackrel{\mbox{\scriptsize{a.s.}}}{\sim}}
\newcommand{\bbeta}{\boldsymbol{\beta}}
\newcommand{\blambda}{\boldsymbol{\lambda}}
\newcommand{\bzeta}{\boldsymbol{\zeta}}
\newcommand{\bLambda}{\boldsymbol{\Lambda}}
\newcommand{\bgamma}{\boldsymbol{\gamma}}
\newcommand{\bdelta}{\boldsymbol{\delta}}
\newcommand{\btheta}{\boldsymbol{\theta}}
\newcommand{\brho}{\boldsymbol{\rho}}
\newcommand{\bmu}{\boldsymbol{\mu}}
\newcommand{\bnu}{\boldsymbol{\nu}}
\newcommand{\bsigma}{\boldsymbol{\sigma}}
\newcommand{\bSigma}{\boldsymbol{\Sigma}}
\newcommand{\bOmega}{\boldsymbol{\Omega}}
\newcommand{\ds}{\displaystyle}
\newcommand{\Ree}{{\rm I}\!{\rm R}}
\newcommand{\Naa}{{\rm I}\!{\rm N}}
\newcommand{\zbar}{\overline{z}}
\newcommand{\wpr}{\mbox{w.pr. }}
\newcommand{\Gstar}{G^\star}
\newcommand{\DP}{\mathcal{DP}}
\newcommand{\DDP}{\mathcal{DDP}}
\newcommand{\DPT}{\mathcal{DPT}}
\newcommand{\PT}{\mathcal{PT}}
\newcommand{\AC}{\mathcal{A}}
\newcommand{\BB}{\mathcal{B}}
\newcommand{\CC}{\mathcal{C}}
\newcommand{\DD}{\mathcal{D}}
\newcommand{\FF}{\mathcal{F}}
\newcommand{\LL}{\mathcal{L}}
\newcommand{\PP}{\mathcal{P}}
\newcommand{\YY}{\mathcal{Y}}
\newcommand{\ZZ}{\mathcal{Z}}
\newcommand{\PY}{\mathcal{PD}}
\newcommand{\law}{\mathcal{L}}
\newcommand{\sigmasib}{\sigma_{\rm si
		\def\ga{\mbox{Ga}}
		\def\iga{\mbox{IGa}}
		\def\siga{\mbox{sqrtIGa}}
		\def\be{\mbox{Be}}
		\def\bin{\mbox{Bin}}
		\def\geo{\mbox{Geo}}
		\def\no{\mbox{N}}
		\def\po{\mbox{Po}}
		\def\pg{\mbox{Pg}}
		\def\st{\mbox{St}}
		\def\sno{\mbox{SkewN}}
		\def\un{\mbox{Un}}
		\def\bga{\mbox{BGa}}
		\def\dir{\mbox{Dir}}
		\def\tri{\mbox{Tri}}
		\def\bebin{\mbox{BeBin}}
		\def\dbb{\mbox{DoubleBB}}
		\def\bbb{\mbox{BBB}}
		\def\E{\mathds{E}}
		\def\V{\mbox{Var}}
		\def\Cov{\mbox{Cov}}
		\def\Cor{\mbox{Corr}}
		\def\p{\mbox{p}}
		\def\b{\mbox{b}}
		\def\P{\mbox{P}}
		\def\Cr{\mbox{Corr}}
		\def\Cv{\mbox{Cov}}
		\def\pt{\mbox{PT}}
		\def\rpt{\mbox{rPT}}
		\def\dpt{\mbox{dPT}}
		\def\bep{\mbox{BeP}}
		\def\d{\mathrm{d}}
		\def\rest{\mbox{rest}}
		\def\data{\mbox{data}}
		\def\tr{\mbox{tr}}
		\def\bc{{\bm c}}
		\def\bk{{\bm k}}
		\def\bM{{\bm m}}
		\def\bt{{\bm t}}
		\def\bu{{\bm u}}
		\def\bv{{\bm v}}
		\def\bx{{\bm x}}
		\def\by{{\bm y}}
		\def\bz{{\bm z}}
		\def\bC{{\bm C}}
		\def\bS{{\bm S}}
		\def\bT{{\bm T}}
		\def\bU{{\bm U}}
		\def\bV{{\bm V}}
		\def\bW{{\bm W}}
		\def\bX{{\bm X}}
		\def\bY{{\bm Y}}
		\def\bZ{{\bm Z}}
		\def\bp{{\bm p}}
		\def\bzero{{\bm 0}}
		\def\simind{\stackrel{\mbox{\scriptsize{ind}}}{\sim}}
		\def\simiid{\stackrel{\mbox{\scriptsize{iid}}}{\sim}}
		\newcommand{\bbeta}{\boldsymbol{\beta}}
		\newcommand{\blambda}{\boldsymbol{\lambda}}
		\newcommand{\bzeta}{\boldsymbol{\zeta}}
		\newcommand{\bLambda}{\boldsymbol{\Lambda}}
		\newcommand{\bgamma}{\boldsymbol{\gamma}}
		\newcommand{\bdelta}{\boldsymbol{\delta}}
		\newcommand{\btheta}{\boldsymbol{\theta}}
		\newcommand{\brho}{\boldsymbol{\rho}}
		\newcommand{\bmu}{\boldsymbol{\mu}}
		\newcommand{\bnu}{\boldsymbol{\nu}}
		\newcommand{\bsigma}{\boldsymbol{\sigma}}
		\newcommand{\bSigma}{\boldsymbol{\Sigma}}
		\newcommand{\bOmega}{\boldsymbol{\Omega}}
		\newcommand{\ds}{\displaystyle}
		\newcommand{\Ree}{{\rm I}\!{\rm R}}
		\newcommand{\Naa}{{\rm I}\!{\rm N}}
		\newcommand{\zbar}{\overline{z}}
		\newcommand{\wpr}{\mbox{w.pr. }}
		\newcommand{\Gstar}{G^\star}
		\newcommand{\DP}{\mathcal{DP}}
		\newcommand{\DDP}{\mathcal{DDP}}
		\newcommand{\DPT}{\mathcal{DPT}}
		\newcommand{\PT}{\mathcal{PT}}
		\newcommand{\AC}{\mathcal{A}}
		\newcommand{\BB}{\mathcal{B}}
		\newcommand{\CC}{\mathcal{C}}
		\newcommand{\DD}{\mathcal{D}}
		\newcommand{\FF}{\mathcal{F}}
		\newcommand{\LL}{\mathcal{L}}
		\newcommand{\PP}{\mathcal{P}}
		\newcommand{\YY}{\mathcal{Y}}
		\newcommand{\ZZ}{\mathcal{Z}}
		\newcommand{\PY}{\mathcal{PD}}
		\newcommand{\law}{\mathcal{L}}
		\newcommand{\sigmasib}{\sigma_{\rm sib}}
		\newcommand{\bpTt}{\tilde{\bm{p}}_{\Theta^*,\boldsymbol{\theta^*}}}
		\newcommand{\bpi}{\boldsymbol \pi}
		\newcommand{\bpT}{\mathbf{\tilde p}_\Theta}}}
\newcommand{\bpTt}{\tilde{\bm{p}}_{\Theta^*,\boldsymbol{\theta^*}}}
\newcommand{\bpi}{\boldsymbol \pi}
\newcommand{\bpT}{\mathbf{\tilde p}_\Theta}

		\newcommand{\kernel}{\mathcal K}

\begin{document}

\baselineskip=24pt

\title{\bf Importance conditional sampling\\ for Pitman-Yor mixtures}
\author{Antonio Canale$^1$, Riccardo Corradin$^2$, and Bernardo Nipoti$^2$\\[2mm]
{\small $^1$ Department of Statistical Sciences, University of Padova, Italy} \\
{\small $^2$ Department of Economics, Management and Statistics, University of Milano Bicocca, Italy} \\[2mm]
{\small {\tt canale@stat.unipd.it, riccardo.corradin@unimib.it  {\rm and} bernardo.nipoti@unimib.it}}}
\date{}

\maketitle

\allowdisplaybreaks

\begin{abstract}
Nonparametric mixture models based on the Pitman-Yor process represent a flexible tool for density estimation and clustering. Natural generalization of the popular class of Dirichlet process mixture models, they allow for more robust inference on the number of components characterizing the distribution of the data. We propose a new sampling strategy for such models, named importance conditional sampling (ICS), which combines appealing properties of existing methods, including easy interpretability and a within-iteration parallelizable structure. An extensive simulation study highlights the efficiency of the proposed method which, unlike other conditional samplers, shows stable performances for different specifications of the parameters characterizing the Pitman-Yor process.  We further show that the ICS approach can be naturally extended to other classes of computationally demanding models, such as nonparametric mixture models for partially exchangeable data. 
\end{abstract}


\section{Introduction}
\label{sec:intro}
Bayesian nonparametric mixtures are flexible models for density estimation and clustering, nowadays a well-established modelling option for applied statisticians \citep{Fru19}.  The first of such models to appear in the literature was the Dirichlet process (DP) \citep{Fer73} mixture of Gaussian kernels by \citet{Lo84}, a contribution which paved the way to  the definition of a wide variety of nonparametric mixture models. In recent years, increasing interest has been dedicated to the definition of mixture models based on nonparametric mixing random probability measures that go beyond the DP \citep[e.g.][]{Nie04,Lij05,Lij05b,Lij07,Arg16}. Among these measures, the Pitman-Yor process (PY) \citep{Per92,Pit95} stands out for conveniently combining mathematical tractability, interpretability, and modelling flexibility \citep[see, e.g.,][]{DeB15}.

Let $\bX=(X_1,\dots,X_n)$ be an $n$-dimensional sample of observations defined on some probability space $(\Omega,\mathscr{A},\P)$ and taking values in $\mathds{X}$, and $\mathscr{F}$ denote the space of all probability distributions on $\mathds{X}$. A Bayesian nonparametric  mixture model is a random distribution taking values in $\mathscr{F}$, defined as
\begin{equation}\label{eq:DPmm}
\tilde{f}(x) = \int_{\Theta} \kernel(x;\theta) \d \tilde p(\theta), 
\end{equation} 
where $\kernel(x;\theta)$ is a kernel and $\tilde p$ is a discrete random probability measure. In this paper we focus on  $\tilde p \sim PY(\sigma, \vartheta; P_0)$, that is we assume that $\tilde p$ is distributed as a PY process with discount parameter $\sigma\in[0,1)$, strength parameter $\vartheta>-\sigma$, and diffuse base measure $P_0 \in \mathscr{F}$. The DP is recovered as a special case when $\sigma=0$. 
Model \eqref{eq:DPmm} can alternatively be written in hierarchical form as
\begin{equation}\label{eq:hier_model}
\begin{split}
X_i \mid \theta_i &\simind \kernel(X_i; \theta_i), \qquad i=1,\ldots,n\\
\theta_i \mid \tilde{p} &\simiid \tilde{p},\\
\tilde p &\sim PY(\sigma, \vartheta; P_0).
\end{split} 
\end{equation} 
The joint distribution of $\btheta=(\theta_1, \dots, \theta_n)$ is characterized by the predictive distribution of the PY, which, for any $i=1,2,\ldots$, is given by
\begin{equation}\label{eq:urn}
\P(\theta_{i+1} \in \d t \mid \theta_{1}, \dots, \theta_{i})= \frac{\vartheta + k_{i} \sigma}{\vartheta + i} P_0(\d t) +  \sum_{j=1}^{k_i} \frac{n_j- \sigma}{\vartheta + i} \delta_{\theta_j^*}(\d t ),
\end{equation} 
where $k_i$ is the number of distinct values $\theta_j^*$ observed in the first $i$ draws and $n_j$ is the number of observed $\theta_{l}$, for $l = 1,\ldots,i$, coinciding with $\theta_j^*$, such that $\sum_{j=1}^{k_i} n_j=i$.

Markov chain Monte Carlo (MCMC) sampling methods represent the gold standard for carrying out posterior inference based on nonparametric mixture models. Resorting to the terminology adopted by \citet{Pap08}, most of the existing MCMC sampling methods for nonparametric mixtures can be classified into marginal and conditional, the two classes being characterized by different ways to deal with the infinite-dimensional random probability measure $\tilde p$. While marginal methods rely on the possibility of analytically marginalizing $\tilde p$ out, the conditional ones exploit suitable finite-dimensional summaries of $\tilde p$.  

Marginal methods for nonparametric mixtures were first devised by \citet{Esc88} and \citet{Esc95}, contributions which focused on DP mixtures of univariate Gaussian kernels. 
Extensions of such proposal include the works of \citet{Mul96}, \citet{Mac94}, \citet{Mac98}, \citet{Nea00}, \citet{Bar13}, \citet{Fav13}, and \citet{Lom17}. It is worth noting that, despite being the first class of MCMC methods for Bayesian nonparametric mixtures appeared in the literature, marginal methods are still routinely used in popular packages such as the \texttt{DPpackage}  \citep{Jar11}, the \emph{de facto} standard software for many Bayesian nonparametric models. 
Alternatively, conditional methods rely on the use of summaries---of finite and possibly random dimension---of realizations of $\tilde p$. To this end, the stick-breaking representation for the PY \citep{Pit97} 
turns out to be very convenient. The almost sure discreteness of the PY allows $\tilde{p}$ to be written as an infinite sum of random jumps $\{p_j\}_{j=1}^{\infty}$ occurring at random locations $\{\tilde\theta_j\}_{j=1}^{\infty}$, that is \begin{equation}
\label{eq:infinitesum}
\tilde{p} = \sum_{j=1}^\infty p_j \delta_{\tilde\theta_j}.
\end{equation}
The distribution of the locations is independent of that of the jumps and, while $\tilde\theta_j \simiid P_0$, the distribution of the jumps is characterized by the following construction:
\begin{align}\label{eq:stick}
p_1 &= V_1, \\
p_j &= V_j\prod_{l=1}^{j-1} (1-V_l), \\
V_j &\simind \mbox{Beta}(1-\sigma, \vartheta + j \sigma).
\end{align} 

A first example of conditional approach can be found in \citet{Ish01} and \citet{Ish02}, contributions that consider a fixed truncation of the stick-breaking representation of a large class of random probability measures, and provide a bound for the introduced truncation error. Along similar lines, \citet{Mul98} and \citet{Arb18} make the truncation level of  the DP and the PY, respectively,  random so to make sure that the resulting error is smaller than a given threshold. 
Exact solutions avoiding introducing truncation errors are the slice samplers of \citet{Wal07} and \citet{Kal11}, the improved slice sampler of \citet{Ge15}, and the retrospective sampler of \citet{Pap08}. It is worth noticing that, although originally introduced for the case of DP mixture models, the ideas behind slice and retrospective sampling algorithms are naturally extended to the more general class of mixture models for which the mixing random probability measure admits a stick-breaking representation \citep{Ish01}, thus including the PY mixture model as a special case. {In this context \citet{Fav13b} propose a general framework for slice sampling the class of mixtures of $\sigma$-stable Poisson--Kingman model.} Henceforth we will use the term slice sampling to refer to the proposals of \citet{Wal07} and \citet{Kal11}, and not to the general definition of slice sampling. 

{Recent contributions have proposed hybrid strategies for posterior sampling nonparametric mixture models, that combine steps of marginal and conditional algorithms and therefore cannot be classified as either type of algorithm. Notable examples are the hybrid sampler of \citet{Lom15} for the general class of Poisson--Kingman mixture models, and the hybrid approach proposed by \citet{Dub20} for a wide range of Bayesian nonparametric models based on completely random measures.}

Marginal methods are appealing for their simplicity and for the fact that the number of random elements that must be drawn at each iteration of the sampler, i.e. the components of $\btheta$, is deterministic and thus bounded. {At the same time, quantifying the posterior uncertainty, e.g. via posterior credible sets, by using the output of marginal methods is in general not straightforward since marginal methods do not generate realizations of the posterior distribution of $\tilde f$, but only of its conditional expectation $\E[\tilde f\mid \btheta,\bX]$, where the expectation is taken with respect to $\tilde p$. To this end, convenient strategies have been proposed, which typically exploit the possibility of sampling approximate realizations of $\tilde p$ conditionally on the values of $\btheta$ generated by the marginal algorithm (see discussions in \citeauthor{Gel02}, \citeyear{Gel02}; \citeauthor{Tad12}, \citeyear{Tad12}; \citeauthor{Arb16}, \citeyear{Arb16}).}
Conditional methods, instead, produce approximate trajectories from the posterior distribution of $\tilde f$, which can be readily used to quantify posterior uncertainty. Moreover, by exploiting the conditional independence of the parameters $\theta_i$'s, given $\tilde p$ or a finite summary of it, conditional methods conveniently avoid sequentially updating the components of $\btheta$ at each iteration of the MCMC, thus 
leading to a fully parallelizable updating step within each iteration. On the other hand, the random truncation at the core of conditional methods such as slice and retrospective samplers makes the number of atoms and jumps that must be drawn at each iteration of the algorithm, random and unbounded. 
By confining our attention to the slice sampler of \citet{Wal07} and, equivalently, its dependent slice-efficient version \citep{Kal11}, we observe that, while its sampling routines are efficient and reliable when the DP case is considered, the same does not hold for the more general class of PY mixtures, specially when large values of $\sigma$ are considered. In practice, we noticed that, even for small sample sizes, the number of random elements that must be drawn at each iteration of the algorithm can be extremely large, often so large to make an actual implementation of the slice sampler for PY mixture models unfeasible. 
It is clear-cut that this limitation represents a major problem as the discount parameter $\sigma$ greatly impacts the robustness of the prior with respect to model-based clustering \citep[see][]{Lij07, Can17biometrics}. In order to shed some light on this aberrant behaviour, we investigate the distribution of the random number $N_n$ of jumps that must be drawn at each iteration of a slice sampler, implemented to carry out posterior inference based on a sample of size $n$. We can define---see  Appendix~\ref{sec:jumps} for details---a data-free lower bound for $N_n$, that is a random variable $M_n$  such that $N_n(\omega)\geq M_n(\omega)$ for every $\omega\in\Omega$ and for every sample of size $n$. $M_n$ is distributed as  $\min\left\{l\geq 1 \,:\, \prod_{j\leq l} (1-V_j)<B_n\right\}$, where the $V_j$'s are defined as in \eqref{eq:stick} and $B_n\sim \text{Beta}(1,n)$: studying the distribution of the lower bound $M_n$ will provide useful insight on $N_n$. Note that, in addition, $M_n$ coincides with the number of jumps to be drawn in order to generate a sample of size $n$ by adapting to the PY case the retrospective sampling idea introduced for the DP by \citet{Pap08}. 
\begin{figure*}
	\begin{center}
		\includegraphics[width=0.8\textwidth]{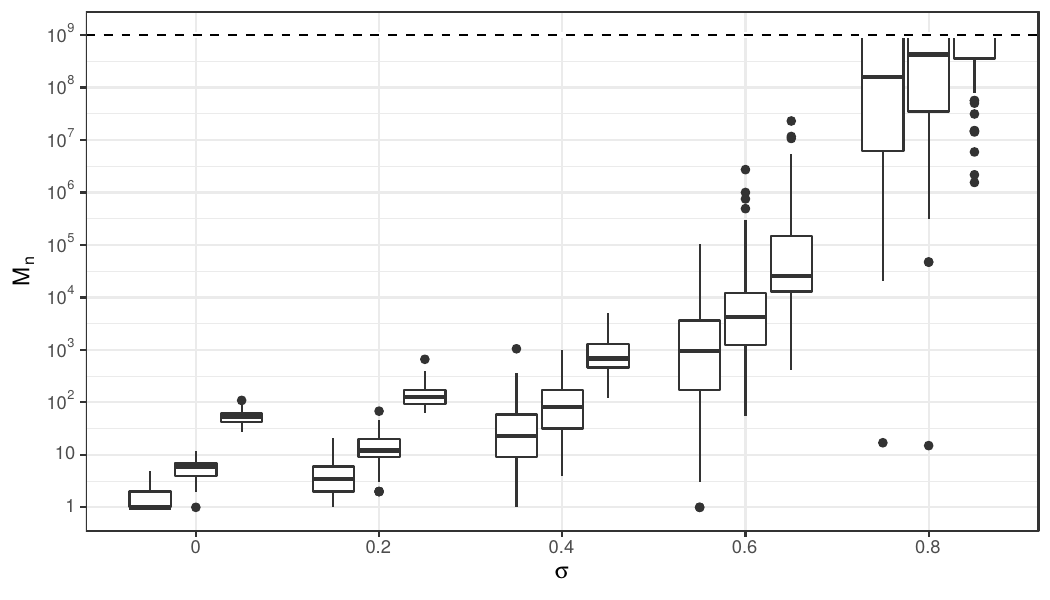}
		\caption{Boxplots for the empirical distributions of $M_n$, with $n=100$, for $\sigma\in\{0,0.2,0.4,0.6,0.8\}$ and different values of $\vartheta$, namely $\vartheta=0.1$ (left), $\vartheta=1$ (middle) and $\vartheta=10$ (right). Results, based on 100 realizations of $M_n$, are truncated at $10^9$ (dashed line).}
		\label{fig1}
	\end{center}
\end{figure*}
Figure~\ref{fig1} shows the empirical distribution of $M_n$, with $n=100$, for various combinations of $\vartheta$ and $\sigma$. 
The estimated median of the distribution of $M_n$ grows with $\sigma$ and, for any given value of $\sigma$, with $\vartheta$. 
It can be appreciated that the size of the values taken by $M_n$, and thus by $N_n$, explodes when $\sigma$ grows beyond $0.5$, fact which leads to the aforementioned computational bottlenecks in routine implementations of the slice sampler. For example, when $\sigma=0.8$, 
the estimated probability of $M_n$ exceeding $10^9$ is equal to $0.35$, $0.42$ and $0.63$, for $\vartheta$ equal to $0.1$, $1$ and $10$, respectively. From an analytic point of view, following \citet{Mul98}, it is easy to show that in the DP case (i.e. $\sigma=0$),  $(M_n -1) \sim \mbox{Poisson}(\vartheta \log(1/B_n))$.
Beyond the DP case (i.e. $\sigma\in(0,1)$), an application of \citet{Arb18} allows us to derive an analogous asymptotic result, which corroborates our empirical findings on the practical impossibility of using the slice sampler for PY mixtures with $\sigma\geq 0.5$. See Proposition~\ref{prop:arbel} and related discussion in the Appendix.

Herein, we propose a new sampling strategy, named importance conditional sampling (ICS), for PY mixture models, which 
combines the appealing features of both conditional and marginal methods, while avoiding their weaknesses, including the computational bottleneck depicted in Figure~\ref{fig1}. Like marginal methods, the ICS 
has a simple and interpretable sampling scheme, reminiscent of Blackwell-MacQueen's P\'olya urn \citep{Bla73}, and 
allows to work with the update of a bounded number of random elements per  iteration; at the same time, being a conditional method,  
it allows for fully parallelizable parameters update and 
it accounts for straightforward approximate posterior quantification. Our proposal exploits the posterior representation of the PY process, derived by \citet{Pit96} in combination with an efficient sampling-importance resampling idea. 
The structure of \citet{Pit96}'s representation makes it suitable for numerical implementations of PY based models, as indicated in \citet{Ish01}, and nicely implemented by  \citet{fallbarat}.  

The rest of the paper is organized as follows. The ICS is described in Section~\ref{sec:approach}. Section~\ref{sec:simul} is dedicated to an extensive simulation study, comparing the performance of the ICS with  state-of-the-art marginal and conditional sampling methods. Section  \ref{sec:app} proposes (reports) an illustrative application, where the proposed algorithm is used to analyse a data set from the Collaborative Perinatal Project \citep{cpp}. In this context, Section \ref{sec:12_hosp} is dedicated to illustrate how the ICS approach can be extended to the case of nonparametric mixture models for partially exchangeable data. Section \ref{sec:conc} concludes the paper with a discussion.  Additional results are presented in the Appendix. 
 
\section{Importance conditional sampling}\label{sec:approach}

The random elements involved in a PY mixture model defined as in \eqref{eq:hier_model} are observations $\bX$, latent parameters $\btheta$ and the PY random probability measure $\tilde p$. The joint distribution of $(\bX, \btheta, \tilde p)$ can be written as 
\begin{equation}\label{eq:joint}
p(\bX, \btheta,\tilde p)
=\prod_{i=1}^n \kernel(X_i;\theta_i) \prod_{j=1}^{k_n} \tilde p(\d \theta_j^*)^{n_j} Q(\tilde p),
\end{equation}
where $\btheta^*=(\theta_1^*, \dots, \theta_{k_n}^*)$ is the vector of unique values in $\btheta$, with frequencies $(n_1, \dots, n_{k_n})$ such that $\sum_{j=1}^{k_n} n_j = n$, and $Q$ is the distribution of $\tilde p \sim PY(\sigma,\vartheta;P_0)$. In line of principle, the full conditional distributions of all random elements can be derived from \eqref{eq:joint} and used to devise a Gibbs sampler. Given that the vector $\bX$, conditionally on $\btheta$, is independent of $\tilde p$, the update of $\btheta$ is the only step of the Gibbs sampler which works conditionally on a realization of the infinite-dimensional $\tilde p$. The conditional distribution $p(\btheta\mid \bX, \tilde p)$ therefore will be the main focus of our attention: its study will allow us to identify a finite-dimensional summary of $\tilde p$, sufficient for the purpose of updating $\btheta$ from its full conditional distribution. As a result, as far as $\tilde p$ is concerned, only the update of its finite-dimensional summary will need to be included in the Gibbs sampler. 
Our proposal exploits a convenient representation of the posterior distribution of a PY process \citep{Pit96},  reported in the next proposition.

\begin{proposition}\label{prop:cor20}\citep[Corollary 20 in][]{Pit96}. Let $t_1,\ldots,t_n\mid\tilde p\sim \tilde p$ and $\tilde p\sim PY(\sigma,\vartheta;P_0)$, and denote by $(t_1^*,\ldots,t_{k_n}^*)$ and $(n_1,\ldots,n_{k_n})$ the set of $k_n$ distinct values and corresponding frequencies in $(t_1,\ldots,t_n)$. The conditional distribution of $\tilde p$, given $(t_1,\ldots,t_n)$, coincides with the distribution of 
	\begin{equation*}
	p_0\tilde q (\cdot)+ \sum_{j=1}^{k_n} p_j \delta_{t_j^*}(\cdot),
	\end{equation*} 
	where 
	$(p_0, p_1,\ldots, p_{k_n})\sim \text{Dirichlet}(\vartheta + k_n\sigma, n_1-\sigma,\ldots,\allowbreak n_{k_n}-\sigma)$ 
	and $\tilde q\sim PY(\sigma,\vartheta+k_{n}\sigma;P_0)$ is independent of $(p_0, p_1, \ldots, p_{k_n})$.
\end{proposition}
In the context of mixture models, Pitman's result implies that the full conditional distribution of $\tilde p$ coincides with the distribution of a mixture composed by a PY process $\tilde q$ with updated parameters, and a discrete random probability measure with $k_n$ fixed jump points at $\bt=(t_1^*,\ldots,t_{k_n}^*)$. This means that, in the context of a Gibbs sampler, while, by conditional independence, the update of each parameter $\theta_i$ is done independently of the other parameters $(\theta_1,\ldots,\theta_{i-1},\theta_{i+1},\ldots, \theta_n)$, the distinct values $\btheta^*$ taken by the parameters at a given iteration, are carried on to the next iteration of the algorithm through $\tilde p$, in the form of fixed jump points $\bt$. 
Specifically, if $\Theta^* = \Theta \setminus \{t^*_1, \dots, t^*_{k_n}\}$, then,   
for every $i = 1, \dots, n$, the full conditional distribution of the $i$-th parameter $\theta_i$ 
can be written as
\begin{equation}\label{eq:fc_thetai}
\P(\theta_i \in \d t \mid X_i, \tilde p) \propto p_0 \kernel(X_i;t)  \tilde q (\d t) + \sum_{j=1}^{k_n} p_j \kernel(X_i;t^*_{j}) \delta_{t^*_{j}}(\d t),
\end{equation}
where $\tilde q$ is the restriction of $\tilde p$ to $\Theta^*$, $p_0=\tilde p(\Theta^*)$ and $p_j=\tilde p(t_j^*)$, for every $j=1,\ldots,k_n$.  
The full conditional in \eqref{eq:fc_thetai} is reminiscent of the Blackwell-MacQueen urn scheme characterizing the update of the parameters in marginal methods: the parameter $\theta_i$ can either coincide with one of the $k_n$ fixed jump points of $\tilde p$ or take a new value from a distribution proportional to $\kernel(X_i;t)\tilde q(\d t)$. The key observation at the basis of the ICS is that, for the purpose of updating the parameters $\btheta$, there is no need to know the whole realization of $\tilde p$ but it suffices to know the vector $\bt$ of fixed jump points of $\tilde p$, the value $\bp=(p_0,p_1,\ldots,p_{k_n})$ taken by $\tilde p$ at the partition $(\Theta^*,t^*_1,\ldots,t^*_{k_n})$ of $\Theta$, and to be able to sample from a distribution proportional to $\kernel(X_i,t)\tilde q(\d t)$. For the latter task, we adopt a sampling-importance resampling approach \citep[see, e.g.,][]{Smi92} with proposal distribution $\tilde q$. It is remarkable that such solution allows us to {approximately} sample from the target distribution 
while avoiding the daunting task of simulating a realization of $\tilde q$ itself. Indeed, for any $m\geq 1$, a vector $\bs=(s_1,\ldots,s_m)$ such that $s_i\mid \tilde q\simiid \tilde q$ can be generated by means of an urn scheme 
exploiting \eqref{eq:urn}. Given the almost sure discreteness of $\tilde q$, the generated vector will show ties with positive probability and thus will feature $r_m\leq m$ distinct values $(s_1^*,\ldots,s_{r_m}^*)$, with frequencies $(m_{1},\ldots,m_{r_m})$ such that $\sum_{j=1}^{r_m}m_j=m$. In turn, importance weights for the resampling step are computed, for any $\ell=1,\ldots,m$, as
\begin{equation*}
 w_\ell=\frac{\kernel(X_i,s_\ell)\tilde q(s_\ell)}{\tilde q(s_\ell)}=\kernel(X_i,s_\ell),
 \end{equation*}
thus without requiring the evaluation of $\tilde q$. 
As a result, the full conditional \eqref{eq:fc_thetai} can be {rewritten} as 
\begin{equation}\label{eq:fc_thetai_bis}
\P(\theta_i \in \d t \mid X_i,\tilde p ) \propto p_0 \sum_{j=1}^{r_m}\frac{m_j}{m}\kernel(X_i;s_j^*)\delta_{s_j^*}(\d t)+ \sum_{j=1}^{k_n} p_j \kernel(X_i;t^*_{j}) \delta_{t^*_{j}}(\d t).
\end{equation}
Once more we highlight an interesting analogy between the conditional approach we propose and marginal methods: the introduction of the auxiliary random variables $s_1^*,\ldots,s_{r_m}^*$ reminds of the augmentation introduced in Algorithm 8 of \citet{Nea00}, marginal algorithm proposed to deal with a non-conjugate specification of the mixture model. From \eqref{eq:fc_thetai_bis} it is straightforward to identify $(\bs,\bt,\bp)$ as a finite-dimensional summary of $\tilde p$, sufficient for the purpose of updating the parameters $\theta_i$ from their full conditionals. This means that, as far as $\tilde p$ is concerned, only its summary $(\bs,\bt,\bp)$ must be included in the updating steps of the Gibbs sampler. To this end, Proposition~\ref{prop:cor20} provides the basis for the update of $(\bs,\bt,\bp)$. Indeed, conditionally on $\btheta$, the fixed jump points $\bt$ coincide with the $k_n$ distinct values appearing in $\btheta$, while the random vectors $\bp$ and $\bs$ are independent with $\bp \sim \text{Dirichlet}(\vartheta+\sigma k_n,n_1-\sigma,\ldots,n_{k_n}-\sigma)$ and the joint distribution of $\bs$ characterized by the predictive distribution of a PY$(\sigma,\vartheta+\sigma k_n;P_0)$, that is, for any $\ell =0,1,\ldots,m-1$,
\begin{equation}\label{eq:pred_s}
\P(s_{\ell+1}\in \d s \mid s_1,\ldots,s_\ell)=\frac{\vartheta+\sigma(k_n+r_\ell)}{\vartheta+\sigma k_n+\ell}P_0(\d s)+\sum_{j=1}^{r_\ell}\frac{m_{j}-\sigma}{\vartheta+\sigma k_n+\ell} \delta_{s^*_j}(\d s),
\end{equation}
where $(s_1^*,\ldots,s_{r_\ell}^*)$ is the vector of $r_\ell$ distinct values appearing in $(s_1,\ldots,s_\ell)$, with corresponding frequencies $(m_{1},\ldots,m_{r_\ell})$ such that $\sum_{j=1}^{r_\ell} m_j=\ell$. 

\SetNlSty{textbf}{[}{]}
\begin{algorithm*}[h!]
	\DontPrintSemicolon
	\textsl{Set admissible initial values $\btheta^{(0)}$}\\
	\For{each iteration $r = 1,\dots,R$}    
	{ 
		\textbf{set} $\bt^{(r)}=\btheta^{*(r-1)};$\\
		\textbf{sample} $\bp^{(r)}$ \textsl{from} 
		$\bp^{(r)} \sim\text{Dirichlet}(\vartheta +  \sigma k_n^{(r-1)},n_{1}^{(r-1)} - \sigma,\ldots,n_{k_n}^{(r-1)} - \sigma);$\\
		\For{each $\ell=0,\ldots,m-1$}   
		{
			\textbf{let} \textsl{$r_\ell^{(r)}$ be the number of distinct values in $(s_1^{(r)},\ldots,s_\ell^{(r)})$}, \textbf{sample} $s_{\ell+1}^{(r)}$ \textsl{from}
			\begin{equation*}
			\P(s_{\ell+1}^{(r)} \in \cdot \mid s^{(r)}_1,\ldots,s^{(r)}_{\ell})=
			\frac{\vartheta+\sigma(k_n^{(r-1)}+r_\ell^{(r)})}{\vartheta+\sigma k_n^{(r-1)}+\ell}P_0(\cdot)+\sum_{j=1}^{r_\ell^{(r)}}\frac{m_j^{(r)}-\sigma}{\vartheta+\sigma k_n^{(r-1)}+\ell} \delta_{s_j^{*(r)}}(\cdot);
			\end{equation*}
		}
		\textbf{let} \textsl{$r_m^{(r)}$ be the number of distinct values in $\bs^{(r)}$};\\ 
\For{each $i=1,\ldots,n$}   
		{
			\textbf{sample} $\theta_i^{(r)}$ \textsl{from}
			\begin{equation*}
			\P(\theta_i^{(r)}= t \mid \cdots )\propto\begin{cases}
			p_0^{(r)} \frac{m_\ell^{(r)}}{m} \kernel(X_i;s_\ell^{*(r)})&\text{ \textsl{if} }t \in \{s_1^{*(r)},\ldots,s_{r_m^{(r)}}^{*(r)}\}\\[6pt]
			p_j^{(r)} \kernel(X_i;t^{*(r)}_{j})&\text{ \textsl{if} }t \in \{t_1^{*(r)},\ldots,t_{k_n^{(r-1)}}^{*(r)}\}\\[6pt]
			0 &\text{ \textsl{otherwise}}
			\end{cases}
			\end{equation*}
		}
		\textbf{let} \textsl{$\btheta^{*(r)} = (\theta_1^{*(r)}, \dots, \theta_{k_n^{(r)}}^{*(r)})$ be the vector of distinct parameters in $\btheta^{(r)}$;}\\
		\For{each $j=1, \dots, k_n^{(r)}$}
		{
			\textbf{let} \textsl{$\mathcal{C}_j^{(r)}$ be the set of indexes $i$ such that} $\theta_i^{(r)}=\theta_j^{*(r)};$\\
			\textbf{update} $\theta_j^{*(r)}$ \textsl{from}
			$\P(\theta_j^{*(r)}\in \d t \mid \cdots ) \propto P_0(\d t) \prod_{i \in \mathcal{C}_j^{(r)}} \kernel(X_i; t);$\\
			
		}
	}
	\textbf{end}
	\caption{ICS for PY mixture model }
	\label{algo:ICSPY}
\end{algorithm*}

By combining the steps just described, as summarized in Algorithm~\ref{algo:ICSPY}, we can then devise a Gibbs sampler which we name ICS. In Algorithm~\ref{algo:ICSPY} and henceforth, the superscript $(r)$ is used to denote the value taken by a random variable at the $r$-th iteration. 
In order to improve mixing, the ICS includes an acceleration step which consists in updating, at the end of each iteration, the distinct values $\btheta^*$ from their full conditional distributions. Namely, for every $j=1,\ldots,k_n$,
\begin{equation}\label{eq:PY_acc}
\P(\theta_j^{*}\in \d t \mid \bX) \propto P_0(\d t) \prod_{i\in C_j} \kernel(X_i;t),
\end{equation}
where $C_j=\{i \in \{1,\ldots,n\} \,:\, \theta_i = \theta_j^*\}$.

Finally, a realization from the posterior distribution of $(\bs,\bt,\bp)$ defines an approximate realization $f$ of the posterior distribution of the random density defined in \eqref{eq:DPmm}, that is
\begin{equation}\label{eq:appr_dens}
\tilde f_m(x)= p_0  \sum_{l=1}^{r_m} \frac{m_l }{m} \kernel(x;s_l^*)+\sum_{j=1}^{k_n} p_j \kernel(x;t_j^{*}).
\end{equation}
	
If the algorithm is run for a total of $R$ iterations, the first $R_b$ of which discarded as burn-in, then the posterior mean is estimated by
\begin{equation*}\hat f(x)= \frac{1}{R-R_b} \sum_{r=R_b+1}^{R} \tilde f_m^{(r)}(x),
\end{equation*}
where $\tilde f_m^{(r)}$ denotes the approximate density sampled from the posterior at the $r$-th iteration. 
The set of densities $\tilde f_m^{(r)}$ can be also used to quantify posterior uncertainty. It is worth remarking though that any such quantification is based on realizations of a finite dimensional summary of the infinite-dimensional $\tilde p$ and thus is, by its nature, approximated. For a quantification of the approximating error one could resort to \citet{Arb18}.

It is instructive to consider how the ICS works for the special case of DP mixture models, that is when $\sigma=0$. In such case, the steps described in Algorithm~\ref{algo:ICSPY} can be nicely interpreted by resorting to three fundamental properties characterizing the DP, namely conjugacy, self-similarity, and availability of finite-dimensional distributions. More specifically, when $\sigma=0$, step 4 of Algorithm~\ref{algo:ICSPY} consists in generating the random weights $\bp$ from a Dirichlet distribution of parameters $(\vartheta,n_1,\ldots,\allowbreak n_{k_n})$. 
This follows by combining the conjugacy of the DP \citep{Fer73}, for which 
\begin{equation*}\tilde p \mid \btheta \sim DP\left(\vartheta + n; \frac{\vartheta}{\vartheta + n}P_0+\sum_{j=1}^{k_n}\frac{n_j}{\vartheta + n} \delta_{\theta_j^*}\right),
\end{equation*}
 with the availability of finite-dimensional distributions of DP \citep{Fer73}, which provides the distribution of $\bp$, defined as the evaluation of the conditional distribution of $\tilde p$ on the partition of $\Theta$ induced by $\btheta$. Moreover, when $\sigma=0$, according to the predictive distribution displayed in step 6 of Algorithm~\ref{algo:ICSPY}, the auxiliary random variables $\bs$ are exchangeable from $\tilde q\sim DP(\vartheta;P_0)$, with $\tilde q$ independent of $\bp$. This is nicely implied by the self-similarity of the DP \citep[see, e.g.,][]{Gho10}, according to which $\tilde q=\tilde p |_{\Theta^*}$ is independent of $\tilde p|_{\Theta\setminus \Theta^*}$, and therefore of $\bp$, and is distributed as a $DP(\vartheta P_0(\Theta^*);P_0|_{\Theta^*})$, and by the diffuseness of $P_0$.
As a result, in the DP case, the auxiliary random variables $\bs$ are generated from the prior model.

\section{Simulation study}\label{sec:simul}

We performed a simulation study to analyze the performance of the ICS algorithm and to compare it with marginal and slice samplers. For the latter, two versions proposed by \citet{Kal11} were considered, namely the dependent and the independent slice-efficient algorithms. The independent version of the algorithm requires the specification of a deterministic sequence $\xi_1,\xi_2,\ldots$, which in our implementation was set equal to $\E[p_1],\E[p_2],\ldots$, with the $p_j$'s defined in \eqref{eq:stick}, in analogy with what was proposed by \citet{Kal11} for the DP (see Algorithm~\ref{algo:SLI2sampler} in the Appendix for more details). All algorithms were written in \texttt{C++} and are implemented in the \texttt{BNPmix} package \citep{bnpmix}, available on CRAN. Aware that different implementations can lead to a biased comparison \citep[see][for an insightful discussion]{Kri17}, we aimed at reducing such bias to a minimum by letting the four algorithms considered here share the same code for most sub-routines.

Throughout this section we consider synthetic data generated from a simple two-component mixture of Gaussians, namely $f_0(x)=0.75 \phi(x; -2.5, 1) + 0.25\phi(x; 2.5, 1)$, with $\phi(\cdot; \mu, \sigma^2)$ denoting the density of a Gaussian random variable with mean $\mu$ and variance $\sigma^2$. All data were analyzed by means of the nonparametric mixture model defined in \eqref{eq:DPmm} and specified by considering a univariate Gaussian kernel $\kernel(x, \theta) = \phi(x; \mu, \sigma^2)$, with $\theta=( \mu, \sigma^2)$, and by assuming a normal-inverse gamma base measure $P_0$ such that $\sigma^2 \sim IG(2,1)$ and $\mu\mid\sigma^2\sim N(0,5\sigma^2)$. Different combinations of values for the parameters $\sigma$ and $\vartheta$, and for the sample size $n$ were considered. The results of this section are then obtained as averages over a specified number of replicates. All algorithms were run for $1\,500$ iterations, of which the first $500$ discarded as burn-in. Convergence of the chains was checked by visual inspection of the trace plots of randomly selected runs, which did not provide any evidence against it. The analysis was carried out by running \texttt{BNPmix} on R 4.0.3 on a 64-bit Windows machine with a 3.4-GHz Intel quad-core i7-3770 processor and 16 GB of RAM.

The first part of our investigation is dedicated to the role of $m$, the size of the auxiliary sample generated for the sampling-importance resampling step within the ICS.  
To this end, we considered two sample sizes, namely $n=100$ and $n=1\,000$, and generated $10$ data sets per size. Such data were then analyzed by considering a combination of values for the PY parameters,  namely $\sigma \in \{0, 0.2, 0.4, 0.6, 0.8\}$ and $\vartheta\in\{1, 10\}$, and by running the ICS with $m \in \{1, 10, 100\}$.
Estimated posterior densities, not displayed here, did not show any noticeable effect of $m$. More interesting findings were obtained when the analysis focused on the quality of the generated posterior sample: larger values for $m$ appear to lead to a better mixing of the Markov chain at the price of additional computational cost. These effects were measured by considering the effective sample size (ESS), computed by resorting to the \texttt{CODA} package \citep{CODA}, and the ratio between runtime, in seconds, and ESS (time/ESS), both averaged over 100 replicates. Following the algorithmic performance analyses of \citet{Nea00}, \citet{Pap08} and \citet{Kal11}, the ESS was computed on the number of clusters---$k_n$ as far as the ICS is concerned---
and on the deviance of the estimated density, with the latter defined as
\begin{equation*}
\mbox{dev}(\bX, \btheta^{(r)}) =-2\sum_{i=1}^n\left(\sum_{j=1}^{k_n^{(r)}} \frac{n_j^{(r)}}{n}\kernel(X_i;\theta_j^{*(r)})\right),
\end{equation*}
for the $r$-th MCMC draw. The ratio time/ESS takes into account both quality of the generated sample and computational cost, and can be interpreted as the average time needed to sample one independent draw from the posterior.

\begin{figure*}[h!]
	\centering
	\includegraphics[width=0.85\textwidth]{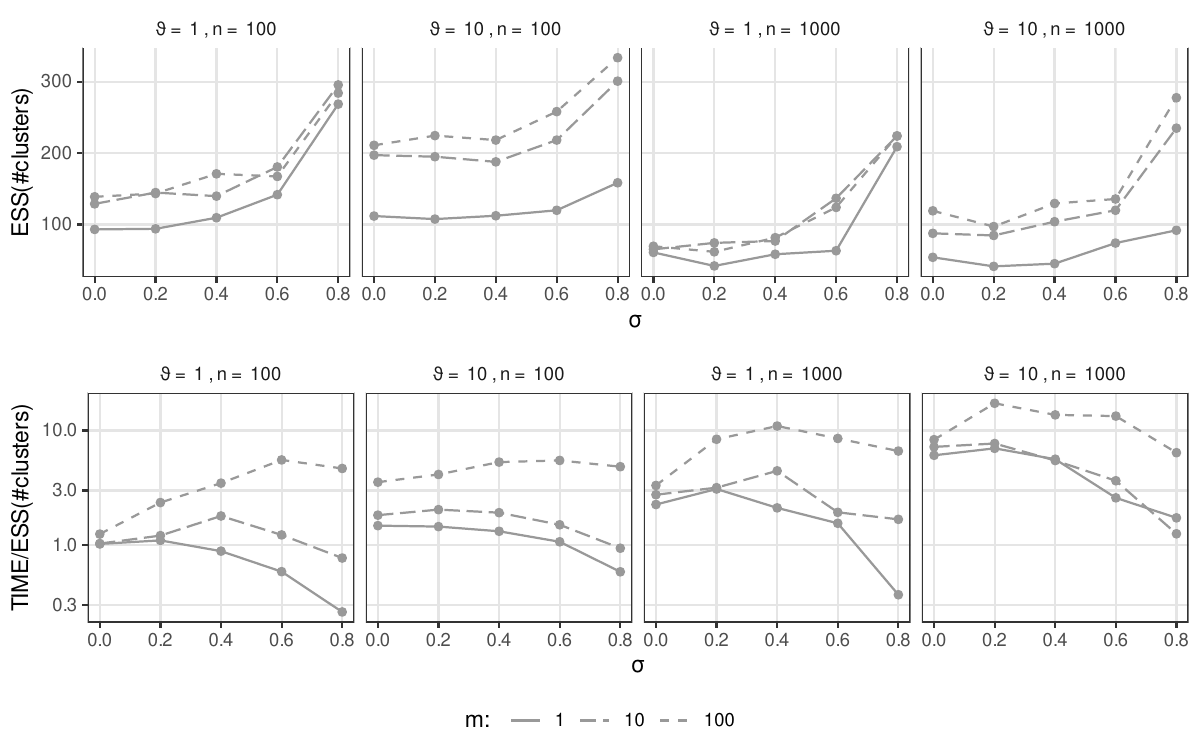}	
	\caption{Simulated data. ICS: ESS computed on the random variable number of clusters (top row) and ratio between runtime (in seconds) and ESS for the same random variable (bottom row). Results are averaged over $100$ replicates.}
	\label{fig:sim1}
\end{figure*}

The results show that larger values of $m$ lead, on average, to a larger ESS, that is to better quality posterior samples. This is displayed in the top row of  Figure~\ref{fig:sim1}, which shows the estimated ESS for $k_n$. We observe that, when averaging over all the considered scenarios, the ESS obtained by setting $m=100$ is $1.82$ and $1.09$ times larger than the average ESS obtained by setting $m=1$ and $m=10$, respectively. At the same time, larger values of $m$ require drawing more random objects per iteration and thus, as expected, lead to longer runtimes. In this sense, the bottom row of Figure~\ref{fig:sim1} clearly indicates that, as far as $k_n$ is concerned, the ratio time/ESS tends to be larger for larger values of $m$. This is particularly evident, 
for example, when $\sigma=0.8$ as the ratio time/ESS corresponding to $m=100$ is, on average, $1.81$ and $1.06$  times larger than the same ratio corresponding to $m=1$ and $m=10$, respectively. Similar conclusions can be drawn by looking at Figure \ref{fig:sim_m_dev}, presented in Appendix \ref{sec:othersimresults} and displaying time/ESS for the deviance of the estimated densities.

When implementing the ICS, the value of $m$ can be tuned based on the desired algorithm performance in terms of quality of mixing and runtime.  As for the rest of the paper, {red}{and for the ease of illustration,} we will work with $m=10$, chosen as a sensible compromise between good mixing and controlled computational cost. 

The second part of the simulation study compares the performance of ICS, marginal sampler, dependent and independent slice-efficient samplers. For the sake of clarity, pseudo-code of the implemented algorithms is provided in Appendix~\ref{ap:competitors}. We considered the sample sizes $n=100$, $n=250$ and $n=1\,000$, and generated 10 data sets per size from $f_0$. These data were then analyzed by considering a combination of values for the PY parameters, namely $\sigma \in \{0, 0.2, 0.4, 0.6, 0.8\}$ and $\vartheta \in \{1, 10, 25\}$.
The results we report are obtained, for each scenario, by averaging over the 10 replicates. 
As for the two slice samplers, due to the aforementioned explosion of the number of drawings per iteration when $\sigma$ takes large values, our analysis was forcefully confined to the case $\sigma\leq 0.4$.
Moreover, the results referring to the case $\sigma=0.4$ are approximate as they were obtained by 
constraining the slice sampler to draw at most $10^5$ components at each iteration: such limitation of our study could not be avoided, given the otherwise unmanageable computational burden associated with this specific setting. Table \ref{tab:bound} in Appendix \ref{sec:othersimresults} shows that such bound was reached more often when large data sets were analyzed. For example, while for $n=100$ the bound was reached on average 12\% and 15\% of the iterations, for independent and dependent slice-efficient samplers respectively, the same happened on average 26\% and 40\% of the iterations when $n=1\,000$. For this reason, these specific results must be considered approximated and, as far as the runtime is concerned, conservative. 
The four algorithms were compared by using the same measures adopted in the first part of the simulation study, namely the ESS for the number of clusters, the ESS for the deviance of the estimated density, and the corresponding ratios time/ESS.

\begin{figure*}[!ht]
	\centering
	\includegraphics[width=0.75\textwidth]{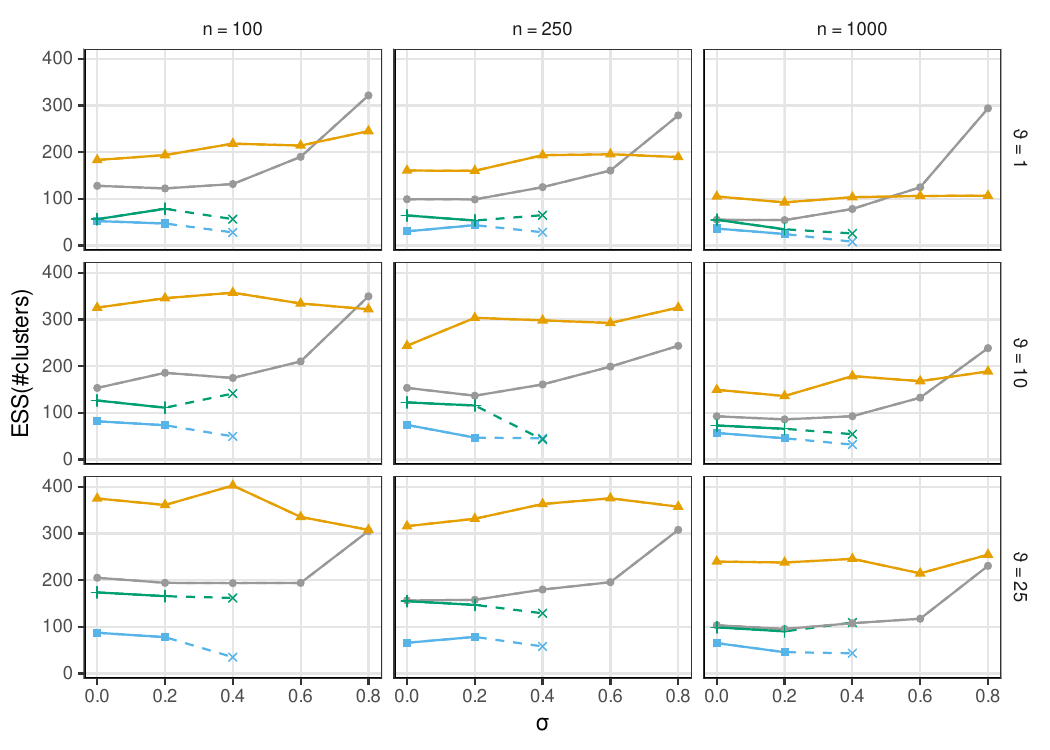}
	\caption{Simulated data. ESS computed on the random variable number of clusters, for ICS (gray), marginal sampler (orange), independent slice-efficient sampler (green) and dependent slice-efficient sampler (blue). Results are averaged over $10$ replicates. The $\times$-shaped marker for the two slice samplers indicates that, when $\sigma=0.4$, the value of the ESS is obtained with an arbitrary upper bound at $10^5$ for the number of jumps drawn per iteration.}
	\label{fig:sim_all_ess}
\end{figure*}
\begin{figure*}[!ht]
	\centering
	\includegraphics[width=0.75\textwidth]{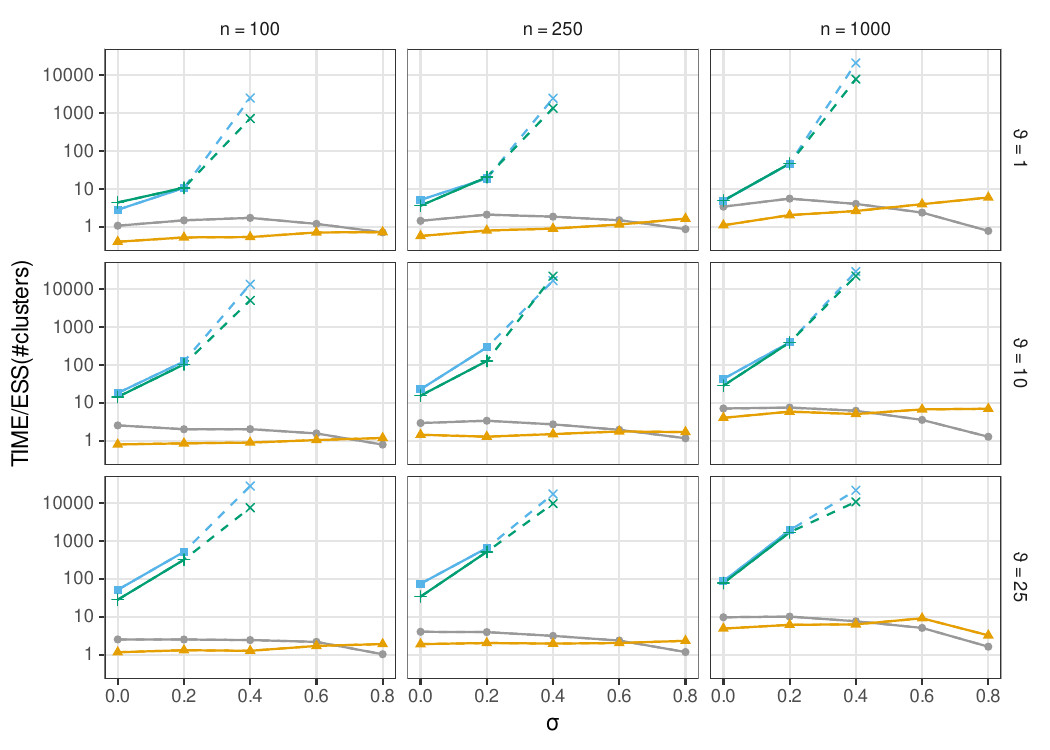} 
	\caption{Simulated data. Ratio of runtime (in seconds) over ESS, in log-scale, computed for the number of clusters, for ICS (gray), marginal sampler (orange), independent slice-efficient sampler (green) and dependent slice-efficient sampler (blue). Results are averaged over $10$ replicates. The $\times$-shaped marker for the two slice samplers indicates that, when $\sigma=0.4$, the value of time/ESS is obtained with an arbitrary upper bound at $10^5$ for the number of jumps drawn per iteration.}
	\label{fig:sim_all_tess}
\end{figure*} 

A clear trend can be appreciated in Figure~\ref{fig:sim_all_ess} where the focus is on the ESS for the number of clusters: the marginal sampler displays, on average, a larger ESS than ICS, whose ESS appears, in turn, uniformly larger than the ones characterizing the two slice samplers. As for the latter two, while the displayed trend is similar, it can be appreciated that the independent algorithm is uniformly characterized by a better mixing.
Results referring to the ratio time/ESS, for the variable number of clusters, are displayed in Figure~\ref{fig:sim_all_tess}. ICS and marginal sampler show in general similar performances. It is interesting to notice though that, while the ratio time/ESS for the marginal algorithm is rather stable over the values of $\sigma$ considered in the study, the same quantity for ICS indicates a slightly better performance when $\sigma$ takes large values.
On the counterpart, the efficiency of the two slice samplers is heavily affected by the value of $\sigma$, with time/ESS exploding when $\sigma$ moves from 0 to 0.4 and when $\vartheta$ is increasing. On the basis of this study, the slice samplers appear competitive options when $\sigma \in\{0,0.2\}$ and a small $\vartheta$ are considered. On the contrary, it is apparent that larger values of $\sigma$ make the two slice samplers less efficient than ICS and marginal sampler. Similar considerations can be drawn when analyzing the performance in terms of deviance of the estimated densities, with the plots for ESS and the ratio time/ESS displayed in Figures \ref{fig:sim1dev} and \ref{fig:sim2dev}. 

\section{Illustrations}\label{sec:app} 

We consider a data set from the Collaborative Perinatal Project (CPP), a large prospective study of the cause of neurological disorders and other pathologies in children in the United States.  Pregnant women were enrolled between 1959 and 1966 when they showed up for prenatal care at one of 12 hospitals.  
While several measurements per pregnancy are available, our attention focuses on two main quantities: the  gestational age (in weeks) and the logarithm of the concentration level of DDE, a persistent metabolite of the pesticide DDT, known to have adverse impact on the gestational age \citep{Lon01}. Our analysis has a two-fold goal. 
First, we focus on estimating and comparing  the joint density of gestational age and DDE for two groups of women, namely smokers and non-smokers.  This will also allow us to assess how the probability of premature birth varies conditionally on the level of DDE. Adopting a nonparametric mixture model will allow us to investigate the presence of clusters within the data. Second, we consider the data set partitioned in the 12 hospitals of the study and focus on the estimation of the hospital-specific distribution of the gestational age, by accounting for possible association across subsamples collected at different hospitals. For this analysis we adopt a nonparametric mixture model for partially exchangeable data and propose an extension of the ICS approach presented in Section \ref{sec:approach}.

\subsection{Cross-hospital analysis}\label{sec:one_hosp}

Smokers and non-smokers groups have sample size of $n_1=1023$ and $n_2=1290$, respectively. For the two groups we independently model the joint distribution of gestational age and DDE by means of a PY mixture model \eqref{eq:hier_model} with bivariate Gaussian kernel function $\kernel(x, \btheta) = \phi(x, \btheta)$, with $\btheta=(\bmu, \bSigma)$, and with conjugate normal-inverse Wishart base measure $P_0 =  N\text{-}IW(\bM_0, k_0, \nu_0, \bS_0)$. In absence of precise prior information on the density to be estimated, we specify a vague base measure following an empirical Bayes approach. Specifically we let $\bM_0$ be equal to the sample average, $\bS_0$ be equal to three times the empirical covariance, $k_0=1/10$, and $\nu_0=5$. These settings are equivalent to assuming that the scale parameter of the generic mixture component coincides with 1.5 times the empirical covariance, while the location parameter is centered on the sample mean with prior variance equal to 10 times the scale parameter. 
Next, we set the parameters $\vartheta$ and $\sigma$ on the basis of the prior distribution they imply on the number of clusters $k_n$, within each group.
Specifically, we set the prior expectation and prior standard deviation for $k_n$ equal to 10 and 20, respectively. 
Our choice implies that a small probability ($\approx 0.05$) is assigned to the event $k_n\geq 50$. 
This argument leads to set $(\sigma,\vartheta)$ equal to $(0.548, -0.485)$ and $(0.5295, -0.4660)$ for the groups of smokers and non-smokers, respectively. The values specified for $\sigma$ are thus larger than $0.5$, situation that is conveniently tackled by the ICS, as displayed by the simulation study of Section \ref{sec:simul}.
An alternative modelling strategy is achieved by introducing a hyperprior distribution for both $\sigma$ and $\theta$. While not explored in this illustration, it is worth stressing that this strategy might be conveniently implemented by adopting the ICS: if the prior on $\sigma$ is defined on $(0,1)$, an implementation of the model requires a sampler whose efficiency is not compromised by the specific values of $\sigma$ explored by the chain.

The analysis of both samples was carried out by running the ICS for $12\,000$ iterations, with the first $7\,000$ discarded as burn-in. Convergence of the chain was assessed as satisfactory by visually investigating the trace plots and by means of the Geweke's diagnostics \citep{Geweke}. Running the analysis of the two samples took less than two minutes in total. It is important to stress that, given the model specification, the same analysis could not be carried out by implementing the slice samplers described in Algorithms~\ref{algo:SLIsampler} and \ref{algo:SLI2sampler}, 
as the value of $\sigma$ would make computation time endless. We could instead implement the marginal sampler described in Algorithm~\ref{algo:MARsampler}, which, as expected, took considerably longer than the ICS (about 11 minutes), due to the moderately large sample sizes.

The contour curves of the estimated joint densities of gestational age and DDE for the two groups are displayed in the left panel of Figure~\ref{fig:sim_biv} and suggest different distributions between smokers and non-smokers, specially when large values for DDE are considered. Differences between the two groups are further highlighted by the right panel of Figure~\ref{fig:sim_biv}, which shows the estimated probability---along with corresponding pointwise 90\% posterior credible bands---of premature birth (i.e. gestational age smaller than 37 weeks), conditionally on the value taken by DDE, for the two groups. Once again, a difference between smokers and non-smokers can be appreciated for large values of DDE, although a sizeable uncertainty is associated with posterior estimates, as displayed by the large credible bands.
\begin{figure*}[!ht]
	\centering
	\includegraphics[width=0.85\textwidth]{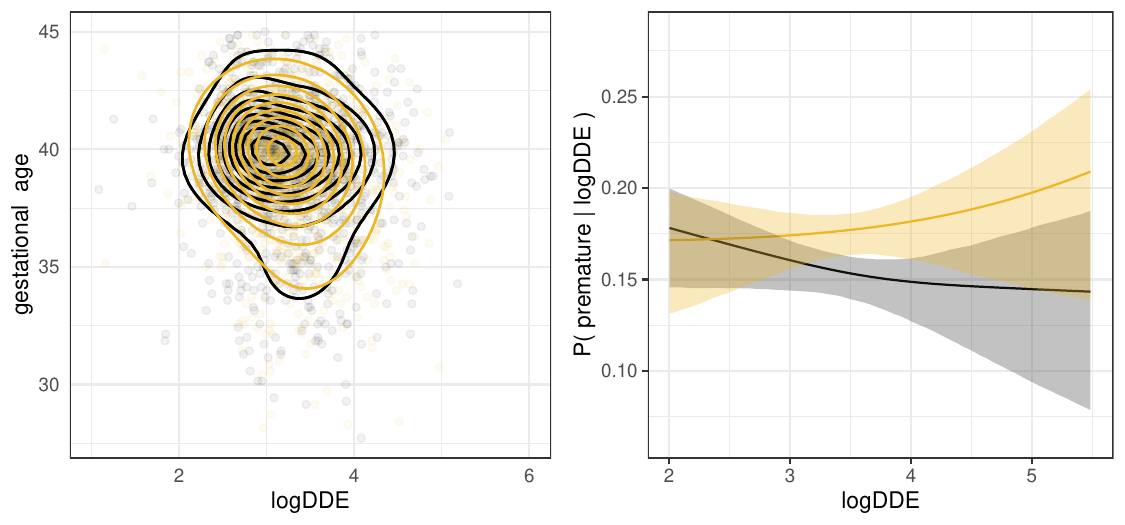} 
	\caption{CPP cross-hospital data. Left: observations and contour curves of the estimated joint posterior density of gestational age and DDE, for smokers (yellow dots and curves) and non-smokers (black dots and curves). 
		Right: estimated probability of premature birth (gestational age below 37 weeks), conditionally on the level of DDE, for smokers (yellow curves) and non smokers (black curves), and associated pointwise  90\% quantile-based posterior credible bands (filled areas).}  
	\label{fig:sim_biv}
\end{figure*} 

\subsection{Multi-hospital analysis}\label{sec:12_hosp}

The same data set as in the previous section is considered here, with observations classified according to both smoking habits of women and the hospitals where they were enrolled. This leads to two samples stratified into $L=12$ strata, with cardinalities summarized by the vectors
\begin{equation*}
\bn_1 =(n_{1,1},n_{2,1},\ldots,n_{12,1})=(236,51,59,38,92,56,67,51,61,187,81,44)
\end{equation*}
and 
\begin{equation*}
\bn_2=(n_{1,2},n_{2,2},\ldots,n_{12,2})=(245, 73, 91, 39, 113, 98, 74, 90, 56, 197, 70, 144)
\end{equation*}
for smokers and non-smokers, respectively. The focus of the analysis is modelling the distribution of gestational age.
 
\subsubsection{A mixture model for partially exchangeable data}

Smokers and non-smokers data are analyzed independently. For each group, heterogeneity across hospitals suggests to assume that data are partially exchangeable in the sense of \citet{deF38}. To account for this assumption, we consider a mixture model for partially exchangeable data, where the stratum-specific mixing random probability measures form the components of a dependent Dirichlet process. Within this flexible class of processes \citep[see][and references therein]{Fot15}, we consider the Griffiths-Milne dependent Dirichlet processes (GM-DDP), as defined and studied in \citet{Lij14a,Lij14b}. For an allied approach see \citet{Gri13}. 
Let  $X_{i,l}$ be the gestational age of the $i$-th woman in the $l$-th hospital, and $\btheta_{l}$ be the vector of latent variables $\theta_{i,l}$ referring to the $l$-th hospital. The mixture model can be represented in its hierarchical form as
\begin{align}\label{mod:ddp}
X_{i,l}\mid \btheta_{1},\ldots,\btheta_{L}&\simind 
\kernel(x_{i,l};\theta_{i,l})\notag\\
\theta_{i_{l},l}\mid (\tilde p_1,\ldots,\tilde p_{L}) &\simiid \tilde p_{l}\\[4pt]
(\tilde p_1,\ldots,\tilde p_{L})&\sim \text{GM-DDP}(\vartheta,z;P_0),\notag
\end{align}
with $l =1,\ldots,L, i = 1,\ldots,n_l\notag$, $\vartheta>0$, $z\in(0,1)$, $P_0$ is a probability distribution on $\R\times\R^+$, and the GM-DDP distribution of the vector $(\tilde p_1,\ldots,\tilde p_{L})$ coincides with the distribution of the vector of random probability measures whose components are defined, for every $l=1,\ldots,L$, as
\begin{equation*}
\tilde p_l=\gamma_{l} \, w_{l}+\gamma_0 \,(1-w_{l}),
\end{equation*}
where $\gamma_1,\ldots,\gamma_{L}\simiid DP(\vartheta z;P_0)$ and $\gamma_0\sim DP(\vartheta(1-z);P_0)$ is independent of $\gamma_l$, for any $l=1,\ldots,L$. Moreover, the vector of random weights $\bw=(w_1, \dots, w_L)$, taking values in $[0,1]^L$, is distributed as a multivariate beta of parameters $(\vartheta z,\ldots,\vartheta z,\vartheta(1-z))$, as defined in \citet{Olk03}, and its components are independent of the random probability measures $\gamma_0,\gamma_1,\ldots,\gamma_L$. As a result, the random probabilities $\tilde p_l$ are, marginally, identically distributed with $\tilde p_l \sim DP(\vartheta;P_0)$ \citep[see][for details]{Lij14a}. 

\subsubsection{ICS for GM-DDP mixture model and its application}\label{sec:ICS_GMDDP}

The ICS can be easily adapted to a variety of models. For example, it naturally fits the partially exchangeable framework of model \eqref{mod:ddp}. 
The ICS algorithm for GM-DDP mixture models is described in Algorithm~\ref{algo:ICS_GMDDP}, and consists of three main steps. First, conditionally on the allocation of observations to clusters referring to either the idiosyncratic process $\gamma_l$, with $l=1,\ldots,L$, or the common process $\gamma_0$, summaries of all the processes, that is $(\bs_{l}, \bt_{l}, \bp_{l})$, for $l=0,\dots,L$, are updated as done in Section~\ref{sec:approach} for a single process, with the proviso that $\sigma=0$. Second, the latent variables $\theta_{i,l}$ are updated for every $l = 1, \dots, L$ and $1 \leq i \leq n_l$; and, third, the components of $\bw$ are sampled. The full conditional distributions for $\theta_{i,l}$ and $\bw$ are provided in Appendix~\ref{app:ICS_GMDDP}.

Model \eqref{mod:ddp} is specified by assuming a univariate Gaussian kernel and normal-inverse gamma base measure $P_0 = N\text{-}IG(0, 5, 4, 1)$. Moreover, the specification $\vartheta=1$ and $z=0.5$ is adopted, with the latter choice corresponding to equal prior weights assigned to idiosyncratic and common components $\gamma_l$ and $\gamma_0$.  The ICS algorithm for the GM-DDP mixture model was run for 10\,000 iterations, the first 5\,000 of which were discarded as burn-in. Estimating posterior densities for smokers and non-smokers, required a total runtime of less than two and a half minutes. Convergence of the chains was assessed by visually investigating the trace plots, which did not provide any evidence against it.  

\begin{figure*}[!ht]
	\centering
	\includegraphics[width=0.88\textwidth]{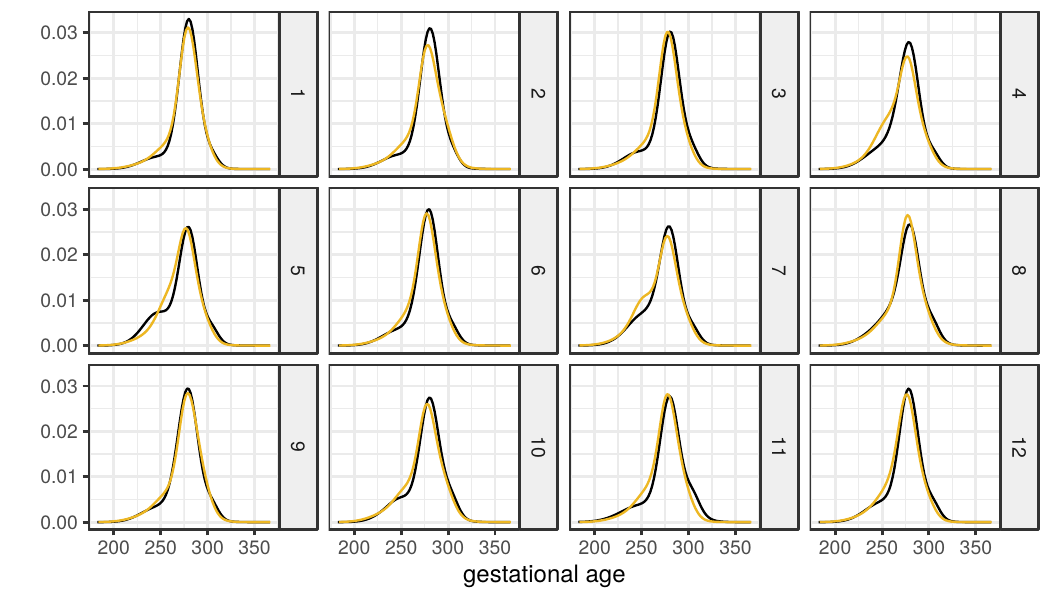}
	\caption{CPP multi-hospital data. Estimated densities of  the gestational age for the 12 hospitals, with comparison between smokers (yellow curves) and non-smokers (black curves).}
	\label{fig:sim_gal}
\end{figure*}

Figure~\ref{fig:sim_gal} shows the estimated densities of the gestational age, for each stratum, with a comparison between smokers and non-smokers. The distribution for smokers is globally more skewed and shifted to the left than the one for non-smokers, indicating an expected more adverse effect of smoke on gestational age.

\section{Discussion}\label{sec:conc}

We proposed a new sampling strategy for PY mixture models, named ICS, which combines desirable properties of existing marginal and conditional methods: the ICS shares easy interpretability with marginal methods, while allowing, likewise conditional samplers, for a parallelizable update of the latent parameters $\btheta$, and for a straightforward quantification of posterior uncertainty. The simulation study of Section~\ref{sec:simul} showed that the ICS overtakes some of the computational bottlenecks characterizing the conditional methods considered in the comparison. Specifically, the ICS can be implemented for any value of the discount parameter $\sigma$, with its efficiency being stable  to the specification of $\sigma$. This is appealing as the discount parameter plays a crucial modelling role when PY mixture models are used for model-based clustering: the ICS allows for an efficient implementation of such models, without the need of setting artificial constraints on the value of $\sigma$.  As far as the comparison of the performances of ICS and other algorithms is concerned, it is important to remark that the independent slice-efficient algorithm proposed by \citet{Kal11} is more general than the one considered in Section~\ref{sec:simul} as other specifications of the deterministic sequence $\xi_1,\xi_2,\ldots$ are possible. As nicely discussed by \citet{Kal11}, the choice of such sequence ``is a delicate issue and any choice has to balance efficiency and computational time''. Alternative specifications of the sequence may be explored on a case-by-case basis but, in our experience, the computational time can be reduced only at the cost of worsening the mixing of the algorithm. It is also worth remarking that the ICS does not rely on any assumption of conjugacy between base measure and kernel, and thus it can be considered by all means a valid alternative to the celebrated Algorithm 8 of \citet{Nea00} when non-conjugate mixture models are to be implemented. Finally, while originally introduced to overtake computational problems arising in the implementation of algorithms for PY mixture models, the idea behind the ICS approach can be naturally extended to other classes of computationally demanding models. As an example, we implemented the same idea to deal with posterior inference based on a flexible class of mixture models for partially exchangeable data. Other extensions are also possible and are currently subject of ongoing research. \\

\noindent \textbf{Acknowledgments.}	The first author is supported by the University of Padova under the STARS Grant. The second and the third authors are grateful to the DEMS Data Science Lab for supporting this work by providing computational resources.

\appendix
\counterwithin{figure}{section}
\counterwithin{table}{section}
\counterwithin{proposition}{section}
\counterwithin{algocf}{section}

\section{On the number of jumps to be drawn with the slice sampler}\label{sec:jumps}
\label{sec:asymptoticexplanation}

Let $N_n$ be the random number of jumps which need to be drawn at each iteration of a slice sampler \citep{Wal07} or, equivalently, its dependent slice-efficient version \citep{Kal11}, implemented to carry out posterior inference based on a sample of size $n$. Conditionally on the cluster assignment variables $c_1,\ldots,c_n$ and on the weights $p_{c_1},\ldots,p_{c_n}$ of the non-empty components of the mixture, $N_n$ is given by
\begin{equation*}
N_n=\min\left\{l\geq 1 \,:\, \sum_{j\leq l} p_j>1-\min(U_1p_{c_1},\ldots,U_np_{c_n})\right\},
\end{equation*}
where the random weights $p_j$'s are defined as in \eqref{eq:stick} and $U_1,\ldots,U_n$ are independent uniform random variables, independent of the weights $p_j$'s. We next define a second random variable $M_n$, function of the same uniform random variables $U_1,\ldots,U_n$, as
\begin{align*}
M_n&= \min\left\{l\geq 1 \,:\, \sum_{j\leq l} p_j>1-\min(U_1,\ldots,U_n)\right\}\\
&\leq \min\left\{l\geq 1 \,:\, \sum_{j\leq l} p_j>1-\min(U_1p_{c_1},\ldots,U_np_{c_n})\right\}=N_n.
\end{align*}
The random number $M_n$ is thus a data-free lower bound for $N_n$, where the inequality $M_n(\omega)\leq N_n(\omega)$ holds for every $\omega \in \Omega$. Studying the distribution of $M_n$ will shed light on the distribution of its upper bound $N_n$. Interestingly, $M_n$ represents also the random number of jumps to be drawn in order to generate a sample of size $n$ from a PY by adapting the retrospective sampling idea of \citet{Pap08}, described in their Section 2 for the DP case. The distribution of $M_n$ coincides with the distribution of $\min\left\{l\geq 1 \,:\, \prod_{j\leq l} (1-V_j)<B_n\right\}$, 
where the stick-breaking variables $(V_j)_{j\geq 1}$ are defined as in \eqref{eq:stick} and $B_n$ is a beta random variable with parameters $1$ and $n$. 
Following \citet{Mul98}, it is easy to show that, when $\sigma=0$, then $M_n-1$ is distributed as a mixture of Poisson distributions, specifically $(M_n-1)\sim \text{Poisson}(\vartheta \log(1/B_n))$. This leads to $\E[M_n]=\vartheta H_n +1$, where $H_{n}=\sum_{l=1}^n l^{-1}$ is the $n$-th harmonic number. It is worth noting that, for $n\rightarrow\infty$, $\E[M_n]\approx \vartheta \log(n)$, that is the growth is logarithmic in $n$, while the contribution of $\vartheta$ is linear. As for the PY process, we resort to \citet{Arb18}, where the asymptotic distribution of the minimum number of jumps of a PY, needed to guarantee that the truncation error is smaller than a deterministic threshold, is studied. We introduce the notation $a_n\simas b_n$ to indicate that $\P(\lim_{n\rightarrow \infty}a_n/b_n=1)=1$ and, by exploiting Theorem 2 in \citet{Arb18}, we prove the following proposition. 

\begin{proposition}\label{prop:arbel}
	Let $M_n=\min\left\{l\geq 1 \,:\, \prod_{j\leq l} (1-V_j)<B_n\right\}$ where the sequence $(V_j)_{j\geq 1}$ is defined as in \eqref{eq:stick} and $B_n$ is a beta random variable with parameters $1$ and $n$. Then, for $n\rightarrow \infty$,
	\begin{equation}\label{eq:limPY}
	M_n-1 \simas \left(\frac{B_n T_{\sigma,\vartheta}}{\sigma}\right)^{-\sigma/(1-\sigma)},
	\end{equation}
	where $T_{\sigma,\vartheta}$, independent of $B_n$, is a polynomially tilted stable random variable \citep{Dev09}, with probability density function proportional to $t^{-\vartheta} f_{\sigma}(x)$, 
	where $f_\sigma$ is the density function of a unilateral stable random variable with Laplace transform equal to $\exp\{-\lambda^\sigma\}$.
\end{proposition}

\begin{proof}
	Define $M(\epsilon)=\min\left\{l\geq 1 \,:\, \prod_{j\leq l} (1-V_j)<\epsilon \right\}$. Following \citet{Arb18},
	\begin{equation}
	M(\epsilon)-1 \simas \left(\frac{\epsilon T_{\sigma,\vartheta}}{\sigma}\right)^{-\sigma/(1-\sigma)},
	\end{equation}
	as $\epsilon\rightarrow 0$. Observe that $M_n=M(B_n)$ and that $B_n\simas 0$ as $n\rightarrow \infty$. We then define the events
	\begin{align*}
	A&=\left\{\omega \in \Omega : M(\epsilon) \not\sim_{\text{a.s.}}\left(\epsilon T_{\sigma,\vartheta}/\sigma\right)^{-\sigma/(1-\sigma)} \text{ as } \epsilon\rightarrow 0\right\}\\
	B&=\left\{\omega \in \Omega : B_n \not\sim_{\text{a.s.}} 0 \text{ as } n \rightarrow \infty \right\}\\
	C&=\left\{\omega \in \Omega : M_n \not\sim_{\text{a.s.}}\left(B_n T_{\sigma,\vartheta}/\sigma\right)^{-\sigma/(1-\sigma)} \text{ as } n\rightarrow \infty\right\}
	\end{align*}
	and observe that $C\subset A\cup B$. Which implies that $\P(C)\leq \P(A\cup B)\leq \P(A)+\P(B)=0$.
\end{proof}

If we define $L_n= \left(B_n T_{\sigma,\vartheta}/\sigma \right)^{-\sigma/(1-\sigma)}$, for any positive integer $n$, the statement of Proposition~\ref{prop:arbel} is tantamount to $M_n-1\simas L_n$ as $n\rightarrow \infty$. The random variable $L_n$ has finite mean if and only if $\sigma\in(0,1/2)$, case in which
$\E[L_n]=c_{\sigma,\vartheta}\Gamma(n+1)/\Gamma(n+2-1/(1-\sigma))$, where 
\begin{equation*}
c_{\sigma,\vartheta}=\sigma^{\sigma/(1-\sigma)}\frac{\Gamma(2-1/(1-\sigma))\Gamma(1+\vartheta/\sigma+1/(1-\sigma))}{\Gamma(\vartheta+1/(1-\sigma))},
\end{equation*}
which implies that $\E[L_n]\approx c_{\sigma,\theta} n^{\sigma/(1-\sigma)}$, when $n\rightarrow \infty$. 
A simple simulation experiment was run to empirically investigate the quality of the asymptotic approximation of $M_n$ provided by $L_n$. The random variable $T_{\sigma,\vartheta}$ appearing in the defintion of $L_n$ was sampled by resorting to \citet{Hof11}. Figure~\ref{fig_Ln_Mn_106} displays the estimated probability of the events $M_n>10^6$ and $L_n>10^6$, as a function of $\sigma\in(0,1)$, for $\vartheta\in\{0.1,1,10\}$ and for different sample sizes $n\in\{100,1\,000,10\,000\}$.

\begin{figure*}[!ht]
\begin{center}
\includegraphics[width=0.9\textwidth]{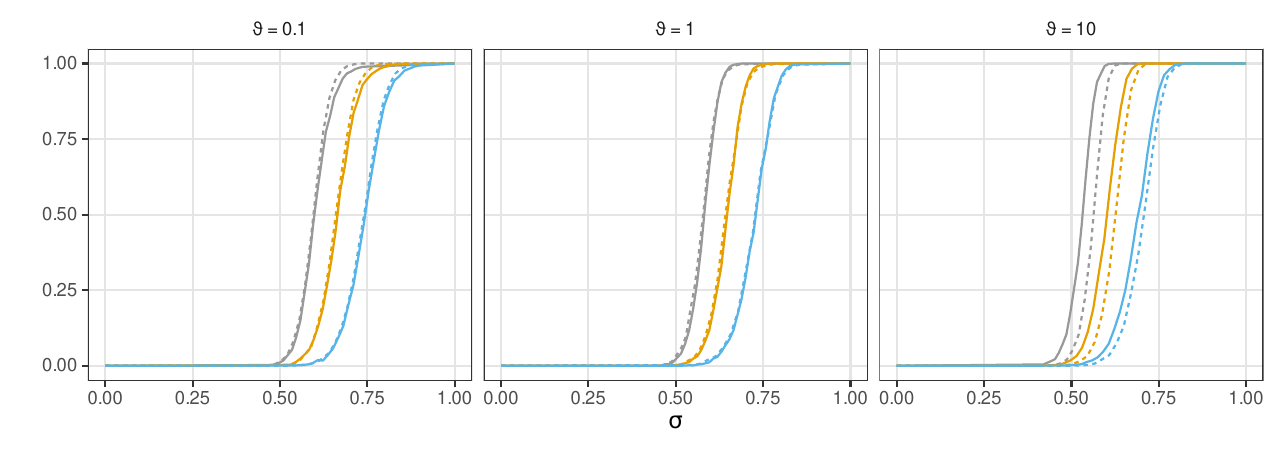}	
\caption{Estimated values for $\P(M_n>10^6)$ (solid curves) and $\P(L_n>10^6)$ (dashed curves) as a function of $\sigma\in(0,1)$, for $n=100$ (blue), $n=1\,000$ (orange), $n=10\,000$ (gray), and for $\theta=0.1$ (left panel), $\theta=1$ (middle panel), $\theta=10$ (right panel).}
	\label{fig_Ln_Mn_106}
\end{center}
\end{figure*}

\section{Additional details on the simulation study}
\label{sec:othersimresults}

This section provides additional results of the simulation study presented in Section \ref{sec:simul}. Table \ref{tab:bound} reports on the number of times the upper bound for the number of jumps drawn at each iteration of dependent and independent slice-efficient samplers was reached. Figures \ref{fig:sim_m_dev}, \ref{fig:sim1dev} and \ref{fig:sim2dev} focus on the functional deviance and display results analogous to those presented in Section \ref{sec:simul} for the random variable number of clusters.

\begin{table}[h!]
	\centering
	\begin{tabular}{llrr}
		&            & \multicolumn{1}{l}{I-SE} & \multicolumn{1}{l}{D-SE} \\ \cline{1-4} 
		\multirow{3}{*}{$\vartheta = 1$}   &  $n=100$     & 0.00 & 0.00\\
		&  $n=250$     & 0.00 & 0.00\\
		&  $n=1000$    & 0.00 & 0.00\\ \hline
		\multirow{3}{*}{$\vartheta = 10$}  &  $n=100$     & 0.06 & 0.07\\
		&  $n=250$     & 0.8 & 0.14\\
		&  $n=1000$    & 0.16 & 0.31\\\hline
		\multirow{3}{*}{$\vartheta = 25$}  &  $n=100$     & 0.30 & 0.39\\
		&  $n=250$     & 0.39 & 0.54\\
		&  $n=1000$    & 0.63 & 0.89\\
		\hline
	\end{tabular}
	\caption{Relative frequency of the of times that the bound $10^5$ on the number of jumps is reached when $\sigma = 0.4$ for the independent slice-efficient algorithm (I-SE) and the dependent slice-efficient algorithm (D-SE).}\label{tab:bound}
\end{table}

\begin{figure*}[h!]
	\begin{center}
		\includegraphics[width=0.9\textwidth]{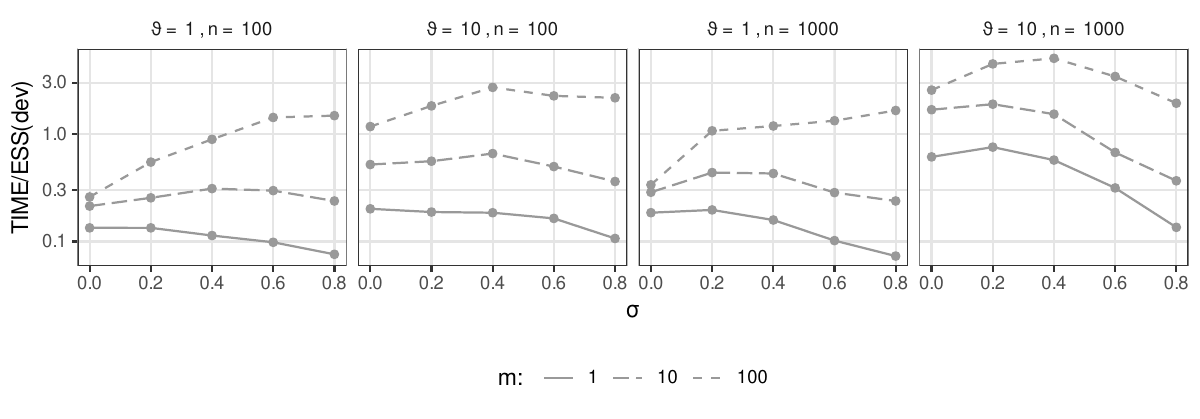}
	\caption{Simulated data. ICS: ratio between runtime (in seconds) and ESS computed on the deviance on a log-scale. Results are averaged over $10$ replicates.}
		\label{fig:sim_m_dev}
	\end{center}
\end{figure*}

\begin{figure*}[h!]
	\begin{center}
		\includegraphics[width=0.75\textwidth]{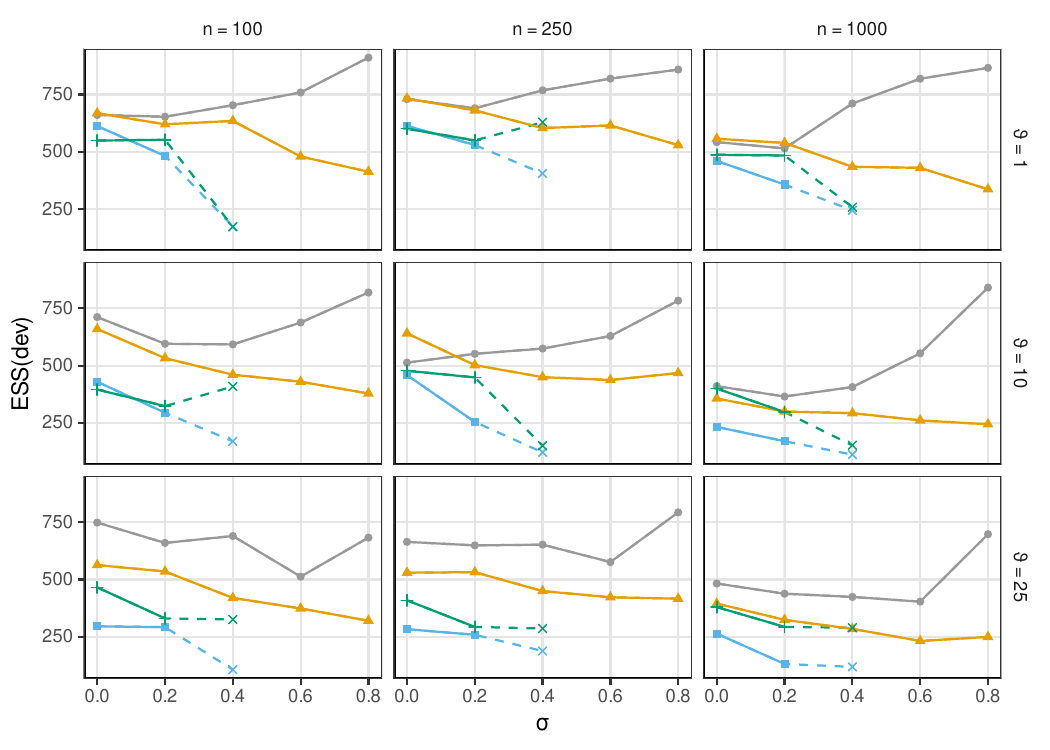}
		\caption{Simulated data. ESS computed on the deviance, for ICS (gray), marginal sampler (orange), independent slice-efficient sampler (green) and dependent slice-efficient sampler (blue). Results are averaged over $10$ replicates. The $\times$-shaped marker for the two slice samplers indicates that, when $\sigma=0.4$, the value of the ESS is obtained with an arbitrary upper bound at $10^5$ for the number of jumps drawn per iteration.}
		\label{fig:sim1dev}
	\end{center}
	\end{figure*}

\begin{figure*}[!ht]
	\begin{center}
		\includegraphics[width=0.75\textwidth]{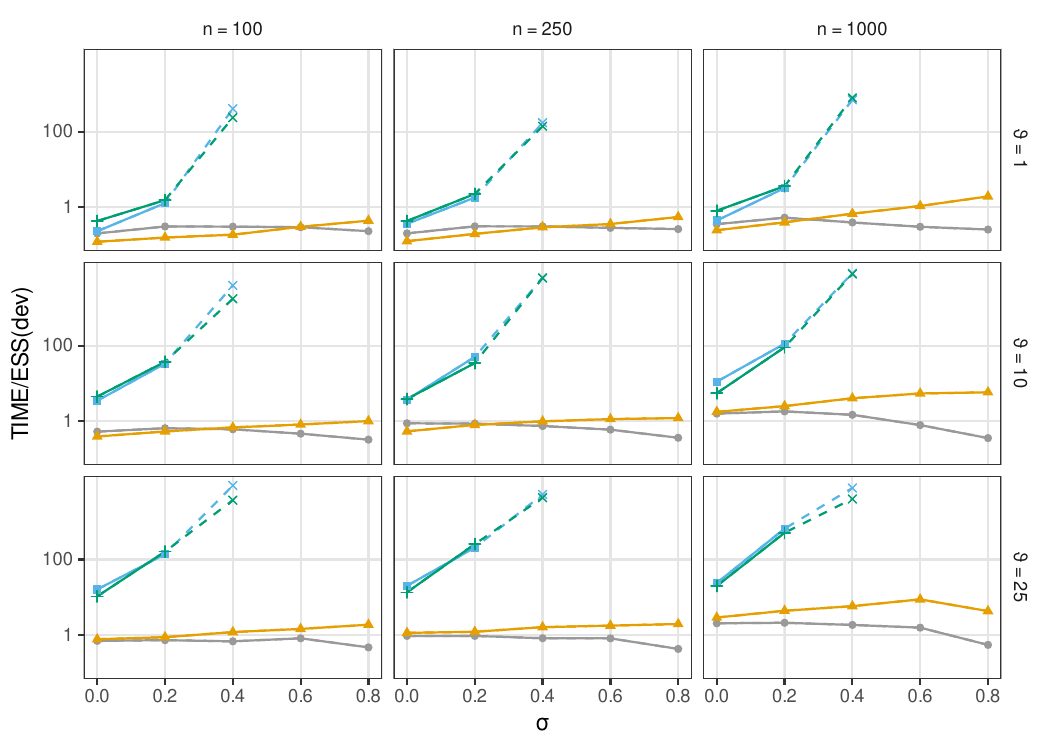}
		\caption{Simulated data. Ratio of runtime (in seconds) over ESS computed on the deviance, in log-scale, for ICS (gray), marginal sampler (orange), independent slice-efficient sampler (green) and dependent slice-efficient sampler (blue). Results are averaged over $10$ replicates. The $\times$-shaped marker for the two slice samplers indicates that, when $\sigma=0.4$, the value of time/ESS is obtained with an arbitrary upper bound at $10^5$ for the number of jumps drawn per iteration.}
		\label{fig:sim2dev}
	\end{center}
\end{figure*}

\section{ICS for GM-DDP}\label{app:ICS_GMDDP}
In order to describe the full conditional distributions of $\theta_{i,l}$ and $\bw$, and to provide the pseudo-code of the ICS for the GM-DDP mixture model, some notation needs to be introduced. Let $r_{m,0}$ and $r_{m,l}$, for $l=1,\ldots,L$, represent the number of distinct values $s_{j,0}^*$ and $s_{j,l}^*$ appearing in the vectors $\bs_0$ and $\bs_l$, respectively. The corresponding frequencies are given by $m_{j,0}$ and $m_{j,l}$, and are such that $\sum_{j=1}^{k_{m,l}}m_{j,l}=m$, for every $l=0,1,\ldots,L$. Let $\btheta_0^*$ be the vector of distinct values appearing in $(\btheta_{1},\ldots,\btheta_{L})$ coinciding with either the $k_{\mathbf{n},0}$ fixed jump points $\bt_0$ of the common process $\gamma_0$ or with any of the $r_{m,0}$ values appearing in $\bs_0$. Similarly, for any $l=1,\ldots,L$, $\btheta_l^*$ denotes the vector of distinct values appearing in $(\btheta_{1},\ldots,\btheta_{L})$ coinciding with either the $k_{\mathbf{n},l}$ fixed jump points $\bt_l$ of the idiosyncratic process $\gamma_l$ or with any of the $r_{m,l}$ values appearing in $\bs_l$. Finally, we let $\mathcal{C}_{j,0}=\{(i,l)\;:\;\theta_{i,l}=\theta^*_{j,0}\}$ and, for $l=1,\ldots,L$, $\mathcal{C}_{j,l}=\{i\;:\;\theta_{i,l}=\theta^*_{j,l}\}$.

The full conditional distribution of $\theta_{i,l}$, for every $l = 1, \dots, L$ and $1 \leq i \leq n_l$, is given, up to a proportionality constant, by\\[-0.7cm]
\begin{multline*}
\P(\theta_{i,l} \in \d t| \ldots ) 
\propto w_l \left(p_{0,l} \sum_{j=1}^{r_{m,l}} \frac{m_{j,l}}{m} \kernel(X_{i,l}, s_{j,l}^*)\delta_{s_{j,l}^*}(\d t) + \sum_{j = 1}^{k_{\bn,l}}  p_{j,l} \kernel(X_{i,l},t_{j,l}^*) \delta_{t_{j,l}^*} (\d t)\right)\\
+ (1 -w_{l}) \left(p_{0,0} \sum_{j=1}^{r_{m,0}} \frac{m_{j,0}}{m} \kernel(X_{i,l}, s_{j,0}^*)\delta_{s_{j,0}^*}(\d t)+ \sum_{j = 1}^{k_{\bn,0}}  p_{j,0} \kernel(X_{i,l},t_{j,0}^*) \delta_{t_{j,0}^*} (\d t)\right).
\end{multline*}

The full conditional for $\bw$ is given, up to a proportionality constant, by\\[-0.5cm]
\begin{multline}\label{eq:ICS-DDP-W}
\P(\bw = (v_1,\ldots,v_L) \mid \dots) \\\propto
\prod_{l = 1}^L \frac{ v_{l}^{ \vartheta z - 1}}{(1 - v_l)^{\vartheta z + 1}} \prod_{i=1}^{n_l}\left(v_l q_{i,l}^{(l)}+(1-v_l)q_{i,l}^{(0)}\right) \left( 1 + \sum_{l=1}^L \frac{v_l}{1 - v_l} \right)^{-L \vartheta z - \theta(1-z)}
\end{multline}
where 
\begin{equation*}
q_{i,l}^{(l)}=p_{0,l} \sum_{j=1}^{r_{m,l}} \frac{m_{j,l}}{m} \kernel(X_{i,l}, s_{j,l}^*)\delta_{s_{j,l}^*}(\d t) + \sum_{j = 1}^{k_{\bn,l}}  p_{j,l} \kernel(X_{i,l},t_{j,l}^*) \delta_{t_{j,l}^*} (\d t),
\end{equation*}
\begin{equation*}
q_{i,l}^{(0)}= p_{0,0} \sum_{j=1}^{r_{m,0}} \frac{m_{j,0}}{m} \kernel(X_{i,l}, s_{j,0}^*)\delta_{s_{j,0}^*}(\d t)+ \sum_{j = 1}^{k_{\bn,0}}  p_{j,0} \kernel(X_{i,l},t_{j,0}^*) \delta_{t_{j,0}^*} (\d t).
\end{equation*}

The pseudo-code of the ICS for the GM-DDP mixture model is presented in Algorithm~\ref{algo:ICS_GMDDP}.

\SetNlSty{textbf}{[}{]}
\begin{algorithm*}
	\DontPrintSemicolon
	\textbf{set} \textsl{admissible initial values for $\btheta_{l}^{(0)}$, for $l=1,\ldots,L$};\\
	\For{each iteration $r = 1,\dots,R$}    
	{ 
		\textbf{set} $\bt_0^{(r)}=\btheta_0^{*(r-1)}$;\\ 		\textbf{sample} $\bp_0^{(r)}$ \textsl{from} $\bp_0^{(r)} \sim\text{Dirichlet}(c (1-z) ,n_{1,0}^{(r-1)} ,\ldots,n_{k_{\mathbf{n},0},0}^{(r-1)})$;\\
		\textbf{sample} $\bs_0^{(r)}$ 
		\textsl{from a} $DP(c(1-z); P_0)$;\\
		\For{each urn $l = 1, \dots, L$}{	
			\textbf{set} $\bt_l^{(r)}=\btheta_l^{*(r-1)}$;\\ 
			\textbf{sample} $\bp_l^{(r)}$ \textsl{from} $\bp_l^{(r)} \sim\text{Dirichlet}(c z ,n_{1,l}^{(r-1)} ,\ldots,n_{k_{\mathbf{n},l},l}^{(r-1)})$;\\
			\textbf{sample} $\bs_l^{(r)}$ 
			\textsl{from a} $DP(cz; P_0)$;\\			
		}
		\textbf{sample} $\bw^{(r)}$ \textsl{from 
			\eqref{eq:ICS-DDP-W}};\\
		\For{each $i=1,\ldots,n_l;\;l=1,\dots,L$}   
		{
			\textbf{sample} $\theta_{i,l}^{(r)}$ \textsl{from}
			\begin{equation*}
			\P(\theta_{i,l}^{(r)}= t \mid \cdots )\propto
			\begin{cases}
			w_l p_{0,l}^{(r)} \frac{m_{j,l}^{(r)}}{m} k(X_{i,l};s_{j,l}^{*(r)})&\text{ \textsl{if} }t \in \{s_{1,l}^{*(r)},\ldots,s_{r_{m,l}^{(r)}}^{*(r)}\}\\[6pt]
			w_l p_{j,l}^{(r)} k(X_{i,l};t^{*(r)}_{j,l})&\text{ \textsl{if} }t \in \{t_{1,l}^{*(r)},\ldots,t_{k_{\bn,l}^{(r-1)},l}^{*(r)}\}\\[6pt]
			(1-w_l)p_{0,0}^{(r)} \frac{m_{j,0}^{(r)}}{m} k(X_{i,l};s_{j,0}^{(r)})&\text{ \textsl{if} }t \in \{s_{1,0}^{*(r)},\ldots,s_{r_{m,0}^{(r)}}^{*(r)}\}\\[6pt]
			(1-w_l)p_{j,0}^{(r)} k(X_{i,l};t^{*(r)}_{j,0})&\text{ if }t \in \{t_{1,0}^{*(r)},\ldots,t_{k_{\bn,0}^{(r-1)},0}^{*(r)}\}\\[6pt]
			0 &\text{ \textsl{otherwise}}
			\end{cases}
			\end{equation*}
		}
		\For{each element $\theta_{j,0}^{*(r)}$ in $\btheta_0^{*(r)}$}
		{
			\vspace{0.1cm}
			\textbf{let} \textsl{$ \mathcal{C}_{j,0}^{(r)}$ be the set of pairs $(i,l)$ such that } $\theta_{i,l}^{(r)}=\theta_{j,0}^{*(r)};$\\
			\textbf{update} $\theta_{j,0}^{*(r)}$ \textsl{from}
			\begin{equation*}\P(\theta_{j,0}^{*(r)}\in \d t \mid \cdots ) \propto P_0(\d t) 
			\prod_{(i,l)\in \mathcal{C}^{(r)}_{j,0}}\kernel(X_{i,l};t);\end{equation*}
		}
		\For{each element $\theta_{j,l}^{*(r)}$ in $\btheta_l^{*(r)}$, $l=1,\dots,L$}
		{
			\vspace{0.1cm}
			\textbf{let} \textsl{$ \mathcal{C}_{j,l}^{(r)}$ be the set of index pairs indexes $i$ such that} $\theta_{i,l}^{(r)}=\theta_{j,l}^{*(r)};$\\
			\textbf{update} $\theta_{j,l}^{*(r)}$ \textsl{from}
			\begin{equation*}\P(\theta_{j,l}^{*(r)}\in \d t \mid \cdots ) \propto P_0(\d t) 
			\prod_{i\in\mathcal{C}_{j,l}^{(r)}} \kernel(X_{i,l};t);\end{equation*}
		}
	}
	\textbf{end}
	\caption{\label{algo:ICS_GMDDP}ICS for GM-DDP mixture model}
\end{algorithm*}

\section{Implementation of the algorithms in \mbox{Section} \ref{sec:simul}}\label{ap:competitors}
This section reports the pseudo-code of marginal, dependent slice-efficient and independennt slice-efficient samplers, algorithms which were implemented for the performance comparison described in Section~\ref{sec:simul}. 
For the sake of simplicity, all the algorithms are described without specifying prior distributions for the hyperparameters. 
Algorithm~\ref{algo:MARsampler} is based on \citet{Esc95}. 
Algorithms~\ref{algo:SLIsampler} and \ref{algo:SLI2sampler} are implemented by following the dependent and independent slice-efficient versions of the slice sampler described in \citet{Kal11}.

\SetNlSty{textbf}{[}{]}
\begin{algorithm}
	\DontPrintSemicolon
	\textbf{set}\textsl{ admissible initial values $\btheta^{(0)}$};\\
	\For{each iteration $r = 1,\dots,R$}    
	{ 
		\For{each $i=1,\ldots,n$}   
		{
			\textbf{let}  $k_{\backslash i}$ \textsl{be the number of distinct values in $\btheta^{(r)}_{\backslash i}=(\theta_1,\ldots,\theta_{i-1},\theta_{i+1},\ldots,\theta_n)$ 
				and $n_j$, for $j = 1, \dots k_{\backslash i}$,  the corresponding frequencies};\\ 
			\vspace{0.1cm}
			\textbf{sample} $\theta_i^{(r)}$ \textsl{from}
			\vspace{-0.3cm}
			\[
			\P(\theta_i^{(r)}=t\mid \dots) \propto
			\begin{cases}
			(n_j - \sigma) \kernel(X_i;\theta_j^{*(r)} )&  \text{\textsl{if} } t = \theta_j^{*(r)}\\
			\qquad\qquad\qquad\qquad\quad\text{ \textsl{and} } &j\in\{1,\ldots,k_{\backslash i}\} 
			\\[8pt]
			(\vartheta + \sigma k_{\backslash i}) \int \kernel(X_i, \theta) P_0(\d \theta)& \text{\textsl{otherwise}}
			\end{cases}
			\]
		}
		
		\For{each unique value $\theta_j^{*(r)}$ in $\btheta^{(r)}$}
		{
			\textbf{let} $\mathcal{C}_j^{(r)}=\{i\in\{1,\ldots,n\}\;:\; \theta_{i}^{(r)}=\theta_{j}^{*(r)}\}$;\\
			\textbf{update} $\theta_j^{*(r)}$ \textsl{from}
			$\P(\theta_j^{*(r)}\in \d t \mid \cdots ) \propto P_0(\d t) \prod_{i \in \mathcal{C}_j^{(r)}} \kernel(X_i; t)$\\
			}
	}
	\textbf{end}
	\caption{\label{algo:MARsampler} Marginal sampler for PY mixture model}
\end{algorithm}

\SetNlSty{textbf}{[}{]}
\begin{algorithm*}
	\DontPrintSemicolon
	\textbf{set}\textsl{ $k=1$, $c_i^{(0)}=1$ for any $i=1,2,\ldots,n$, and an admissible initial value $\tilde{\theta}_1^{(0)}$};\\
	\For{each iteration $r = 1,\dots,R$}    
	{ 
		\For{each $i = 1, \ldots, n$}{
			\textbf{sample} 
			$u_i \sim \mbox{Unif}([0, p_{c_i^{(r-1)}}])$;
		}
		
		\While{$\sum_{j = 1}^{k} w_j < 1 - u_i$, for any $i$ }
		{
			\vspace{0.1cm}
			\textbf{sample} \textsl{a new weight} $v_{k+1} \sim \mbox{Beta}(1 - \sigma, \vartheta + (k+1)\sigma)$;\\
			\textbf{set} $p_{k+1} = v_{k+1} \prod_{l < k+1} (1 - v_l)$;\\
			\textbf{sample} $\tilde{\theta}_{k+1}^{(r)} \sim P_0;$\\
			\textbf{set} $k = k + 1$;
			
		}
		\For{each $i=1,\ldots,n$}   
		{
			\textbf{sample} $c_i^{(r)}$ \textsl{from}
			\begin{align*}
			&\P(c_i^{(r)}= j \mid \cdots )
			\propto
			\begin{cases}
			\mathds{1}_{[p_j > u_i]} \kernel(X_i, \tilde{\theta}_j^{(r)}) &\text{ \textsl{if} }j \in \{1, \dots, k\}
			\\[8pt]
			0 &\text{ \textsl{otherwise}}
			\end{cases}
			\end{align*}
			
		}
		\textbf{set} $k=max(c_1, \dots, c_n)$;\\
		\For{each $j = 1, \dots k$}
		{
			\textbf{let} $\mathcal{C}_j^{(r)}$ \textsl{be the set of indexes having $c_i^{(r)}=j$};\\
			\textbf{update} $\tilde{\theta}_j^{(r)}$ \textsl{from}
			$\P(\tilde{\theta}_j^{(r)}\in \d t \mid \cdots ) \propto P_0(\d t) \prod_{i \in \mathcal{C}_j^{(r)}} \kernel(X_i; t);$\\
			\vspace{0.1cm}
			\textbf{sample} 
			$v_j \sim \mbox{Beta}\left(1 - \sigma + n_j, \vartheta + j\sigma + n_j^+ \right)$,\\ 
			\textsl{where $n_j$ is the cardinality of $\mathcal{C}_j^{(r)}$ and $n_j^+ = n - \sum_{l=1}^j n_j$};\\
			\vspace{0.1cm}
			\textbf{set} $p_j = v_j \prod_{l<j}(1 - v_l)$;
		}
	}
	\textbf{end}
	
	\caption{\label{algo:SLIsampler}Dependent slice-efficient sampler for PY mixture model}
\end{algorithm*}

\SetNlSty{textbf}{[}{]}
\begin{algorithm*}
	\DontPrintSemicolon
	\textbf{set}\textsl{ $k=1$, $c_i^{(0)}=1$ for any $i=1,2,\ldots,n$, and an admissible value $\tilde{\theta}_1^{(0)}$};\\
	\textbf{set}  $\xi_1=(1-\sigma)/(\vartheta+1)$;\\
	\For{each iteration $r = 1,\dots,R$}    
	{ 
		\For{each $i = 1, \ldots, n$}{
			\textbf{sample} 
			$u_i \sim \mbox{Unif}([0,\xi_{c_i^{(r-1)}}])$;
		}
		
		\While{$\sum_{j = 1}^{k} \xi_j < 1 - u_i$, for any $i$ }
		{
			\vspace{0.1cm}
			\textbf{set} $\xi_{k+1} = \xi_k \left(\vartheta + k \sigma\right)/\left(\vartheta + 1 + k \sigma\right)$;\\
			\textbf{sample} \textsl{a new weight} $v_{k+1} \sim \mbox{Beta}(1 - \sigma, \vartheta + (k+1)\sigma)$;\\
			\textbf{set} $p_{k+1} = v_{k+1} \prod_{l < k+1} (1 - v_l)$;\\
			\textbf{sample} $\tilde{\theta}_{k+1}^{(r)} \sim P_0;$\\
			\textbf{set} $k = k + 1$;
			
		}
		\For{each $i=1,\ldots,n$}   
		{
			\textbf{sample} $c_i^{(r)}$ \textsl{from}
			\begin{align*}
			&\P(c_i^{(r)}= j \mid \cdots ) \\
			&\propto
			\begin{cases}
			\mathds{1}_{[\xi_j > u_i]} \frac{p_j}{\xi_j} \kernel(X_i, \tilde{\theta}_j^{(r)}) &\text{ \textsl{if} }j \in \{1, \dots, k\}
			\\[8pt]
			0 &\text{ \textsl{otherwise}}
			\end{cases}
			\end{align*}
			
		}
		\textbf{set} $k=max(c_1, \dots, c_n)$;\\
		\For{each $j = 1, \dots k$}
		{
			\textbf{let} $\mathcal{C}_j^{(r)}$ \textsl{be the set of indexes having $c_i^{(r)}=j$};\\
			\textbf{update} $\tilde{\theta}_j^{(r)}$ \textsl{from}
			$\P(\tilde{\theta}_j^{(r)}\in \d t \mid \cdots ) \propto P_0(\d t) \prod_{i \in C_j^{(r)}} \kernel(X_i; t);$\\
			\vspace{0.1cm}
			\textbf{sample} 
			$v_j \sim \mbox{Beta}\left(1 - \sigma + n_j, \vartheta + j\sigma + n_j^+ \right)$,\\ 
			\textsl{where $n_j$ is the cardinality of $\mathcal{C}_j^{(r)}$ and $n_j^+ = n - \sum_{l=1}^j n_j$};\\
			\vspace{0.1cm}
			\textbf{set} $p_j = v_j \prod_{l<j}(1 - v_l)$;
		}
	}
	\textbf{end}
	
	\caption{\label{algo:SLI2sampler}Independent slice-efficient sampler for PY mixture model (with $\xi_j=\E[p_j]$ for $j=1,2,\ldots$)} 
	\end{algorithm*}

\clearpage

\bibliographystyle{apalike}
\begin{small}
\bibliography{ICSbib}
\end{small}

\end{document}